\documentclass[a4paper,10pt]{article}

\usepackage[top=1in, bottom=1.25in, left=1.3in, right=1.3in]{geometry}
\usepackage{graphicx}
\usepackage{authblk}
\usepackage[round]{natbib}
\usepackage{hyperref}
\usepackage{amssymb,amsmath}
\usepackage{physics}
\usepackage{mathtools}
\usepackage{lipsum}
\usepackage{amsthm}
\usepackage{tikz}
\usepackage{kbordermatrix}
\usepackage{array,booktabs}

\usepackage{palatino}

\newtheorem{prop}{Proposition}
\newtheorem{lemm}{Lemma}
\newtheorem{cor}{Corollary}[prop]

\begin{document}

\title{Variational inequalities and mean-field approximations for partially observed systems of queueing networks.\thanks{Work supported by RCUK through the Horizon Digital Economy Research grants (EP/G065802/1, EP/M000877/1).
}}

\author[a]{Iker Perez\thanks{Corresponding author address: Horizon Digital Economy Research, Triumph Road, Nottingham, NG7 2TU. Email: iker.perez@nottingham.ac.uk}}
\author[b]{Giuliano Casale}
\affil[a]{School of Mathematical Sciences, University of Nottingham}
\affil[b]{Department of Computing, Imperial College London}
\date{\vspace{-25pt}}

\maketitle

\begin{abstract}
Queueing networks are systems of theoretical interest that find widespread use in the performance evaluation of interconnected resources. In comparison to counterpart models in genetics or mathematical biology, the stochastic (jump) processes induced by queueing networks have distinctive \textit{coupling} and \textit{synchronization} properties. This has prevented the derivation of variational approximations for conditional representations of transient dynamics, which usually rely on simplifying independence assumptions. In this paper, we present a model augmentation to a multivariate counting process for interactions across service stations, and we first enable the variational evaluation of mean-field measures for partially-observed open and closed multi-class networks. Finally, we show that our approximating framework offers a viable, efficient and improved alternative for inference and uncertainty quantification tasks, where existing variational or numerically intensive solutions do not work.
\end{abstract}

\vspace{10pt}

\section{Introduction}

\textit{Queueing networks} (QNs) are systems of theoretical and practical interest in the design of computing systems \citep{kleinrock1976queueing}, as well as in the optimization of business processes arising in factories, shops, offices or hospitals \citep{buzacott1993stochastic,Koole2002,Osorio2009996}. They are formed by interconnected resources routing and processing jobs, and their behaviour often gives rise to complex families of stochastic (jump) processes. In applications, they provide the means to assess modifications, diagnose performance and evaluate robustness in multiple service infrastructures. 

Formally, QNs are associated with \textit{coupled} or \textit{synchronized} (Markov) jump processes. Here, every change in a marginal population count (jobs within a queue) is triggered by an arrival (or departure) from an additional resource. The multivariate behaviour across populations in the underlying jump model is thus strongly interlinked; preventing the derivation of variational approximations for transient dynamics that rely on simplifying independence assumptions \citep[cf.][]{opper2008variational,cohn2010mean}. In this paper, our main contribution is to present a complete probabilistic (hierarchical) formulation of open and closed networks, and to first enable the variational evaluation of approximating mean-field measures for such partially-observed \textit{coupled} systems. Additionally, we discuss the relation to analogue tasks in domains such as genetics or mathematical biology, and present use cases of our results within uncertainty quantification and Bayesian inferential tasks, applied to examples where existing MCMC/variational solutions either (i) do not scale well or (ii) are unusable. The results within this paper are relevant for single or multi-class Markovian systems (and related stochastic models in genetics or biology), with either finite or infinite processors, multiple types of service disciplines and probabilistic routings.

\noindent\textbf{Motivation.} The quantitative basis for the evaluation of a networked system is a set of estimates for the service requirements in the resources. To that end, a foundational inferential study begins with a set of measurements (queue lengths, visit counts, response times, \dots) along with an associated likelihood function interrelated with service rates and the underlying stochastic (jump) dynamics. However, the measurements often provide little indirect information \citep{sutton2011}, and there exist strong impediments to integrate over uncertainty in the jump process trajectories \citep{armero1994assessment,perez2017auxiliary}. Recently proposed techniques in \cite{sutton2011,perez2017auxiliary} are only relevant for reduced types of systems, and are sustained on intense \textit{Markov Chain Monte Carlo} sampling procedures. Thus, they suffer from scalability problems associated with complex multivariate temporal dependencies \citep{bobbio2008analysis}, as a result of the aforementioned \textit{synchronization} properties, along with job priorities or the existence of feedback loops. This differentiates the network evaluation problem from analogue statistical tasks for jump processes associated with mathematical biology \citep{hobolth2009} or genetics \citep{golightly2015bayesian}. Currently, practical solutions often rely on steady-state metrics \citep{kraft2009estimating} or end-to-end measurements \citep{Liu200636}; thus, the effects of system uncertainty are not understood \citep[see][for a review]{Spinner201551}. In this paper, we show that a variational framework targeted at conditional representations of transient dynamics offers a viable (and efficient) alternative to existing numerically intensive solutions presented in \cite{sutton2011,perez2017auxiliary}, in order to enable foundational inferential and uncertainty quantification tasks with QNs and their underlying jump process representations. 

\noindent\textbf{Structure.} The rest of the paper is organised as follows.  In Section \ref{introSec} we offer a (probabilistic) hierarchical formulation of a queueing system along with the problem statement. Section \ref{overSection} introduces an approximating network model and offers a summary of the main results to be presented later in the paper. Sections \ref{countModelSec} and \ref{FunctionalSec} include the main contributions in our work; these discuss the treatment of the network system by means of interactions in network resources, and further present the results, proofs and technical details that contribute to later algorithmic constructions. In Section \ref{exmpls}, we guide the reader through applications of our results within inferential and network evaluation tasks and in Section \ref{discussion} we conclude the paper with a discussion.

\section{Queueing systems and jump processes} \label{introSec}

In the following, we employ shorthand notation for densities, base measures and distributions whenever these are clear from the context.  From here on, let $(\Omega,\mathcal{F})$ denote a measurable space with the regular conditional probability property, supporting the various rates, trajectories and observations. A general form queueing network comprises some $M\in\mathbb{N}$ service stations along with a \textit{set of job classes} $\mathcal{C}$. The stations are connected by a \textit{network topology} that governs the underlying routing mechanisms; when a job is serviced in a station, it can either queue for service at a different node, or depart the network. Such topology is often defined as a set of \textit{routing probability matrices} $\{P^c\}_{c\in\mathcal{C}}$, with elements $p^c_{i,j}$ that denote the probability for a class $c\in\mathcal{C}$ job to immediately transit to queueing station $j$ after service completion in station $i$, for all $0\leq i,j \leq M$. In open queueing systems, the index $0$ is used as a \textit{virtual} external node that represents the source and destination of job arrivals and departures to, and from, the network. In closed systems, this index may either not exist, or instead refer to a delay server that routes departing jobs back into the network. Also, it holds that $\sum_{j=0}^M p^c_{i,j} = 1$, for all $0\leq i\leq M, c\in\mathcal{C}$. 

We address time-homogeneous Markovian systems that are parametrized by exponential inter-arrival and service times, with non-negative rates $\boldsymbol{\mu} = \{ \mu^c_i \in\mathbb{R}_+  : 0\leq i\leq M, c\in\mathcal{C}\}$, which may vary across service stations and job classes. The servers in the network stations may have finite or infinite processors, and service disciplines can differ across a range of \textit{processor sharing} (PS) policies, \textit{first-come first-served} (FCFS) and variations including \textit{last-come first-served} (LCFS) or \textit{random order} (RO) nodes. In some cases, FCFS processors may require shared processing times across the various job classes (cf. \cite{baskett1975open}). For simplicity and ease of notation, class \textit{switching}, service \textit{priorities} or queue-length dependent service rates are not discussed in detail, however, these follow naturally and we later present some examples of such instances. Under standard exponential service assumptions, the underlying system behaviour is described by an MJP $X=(X_t)_{t\geq 0}$ with values defined in a measurable space $(\mathcal{S},\mathcal{P}(\mathcal{S}))$. Here, $\mathcal{S}$ denotes a countable set of feasible states in the network, usually infinite in open or mixed systems and finite within closed ones; $\mathcal{P}(\mathcal{S})$ denotes the power set of $\mathcal{S}$. We allow for $\mathcal{S}$ to support vectors of integers that represent job counts across the various class types and service nodes, and denote by $X_t^{i,c}$ the number of class $c$ jobs in station $i>0$ at time $t\geq 0$. Note that here we ignore the loads in delay nodes $(i=0)$ within closed systems, since these are uniquely determined given the number of jobs in the remaining stations. The \textit{infinitesimal generator matrix} $Q$ of $X$ is such that
$$\mathbb{P}(X_{t+\mathrm{d}t}=x'|X_{t}=x) = \mathbb{I}(x=x') + Q_{x,x'}\mathrm{d}t + o(\mathrm{d}t)$$
for all $x,x'\in\mathcal{S}$. This can be an infinite matrix, it is generally sparse and its entries describe rates for transitions across states in $\mathcal{S}$. Rows in $Q$ must sum to $0$ so that $Q_{x,x'}\geq 0$ for all $x \neq x'$, and $Q_x \coloneqq Q_{x,x} = - \sum_{x'\in\mathcal{S}: x\neq x'} Q_{x,x'}$. 

Hence, jumps in the process $X$ are caused by jobs being routed through nodes in the underlying network model. We often say that a state $x'\in\mathcal{S}$ is \textit{accessible} from $x\in\mathcal{S}$, and write $x\xrightarrow{i,j,c} x'$ for its corresponding jump, if $x'$ may be reached from $x$ by means of a class-$c$ job transition between the stations $i$ and $j$, in the direction $i\rightarrow j$. We further denote $$\mathcal{T}=\{(i,j,c)\in\{0,\dots,M\}^2\times\mathcal{C}: p_{i,j}^c>0\}$$ for the finite set of all \textit{feasible} job transitions in the system, and we remark that the generator $Q$ of $X$ is populated by some positive real-valued rates $\boldsymbol{\lambda}= \{ \lambda_{\boldsymbol{\eta}} \in\mathbb{R}_+ \, : \, {\boldsymbol{\eta}}\in\mathcal{T}\} $ that define the intensities for these job routings, with $\lambda_{i,j,c} = \mu^c_{i}\cdot p^c_{i,j}$ for all $(i,j,c)\in\mathcal{T}$. 

In Figure \ref{exmplQueues} we observe diagrams that illustrate this notation in an open single-class network. On the left, we see $3$ stations with different rates, disciplines and server counts. The topology $P$ is such that $|\mathcal{T}|=5$ and $p_{0,1}=1-p_{0,2}\in (0,1)$, $p_{1,3}=p_{2,3}=p_{3,0}=1$ ($p_{i,j}=0$ otherwise). On the right, we find the corresponding job transition rates across the $4$ pairs of connected nodes. In this single-class example, $X$ monitors counts across the stations s.t. $X_t=(X^{1}_t,X^{2}_t,X^{3}_t)\in\mathcal{S}$ for all $t\geq 0$; also, the generator $Q$ is an infinite matrix with $\textstyle Q_{x,x'}=\lambda_{i,j}\cdot (K_i \wedge x_i )$ for all pairs $x,x'\in\mathcal{S}$ with  associated transition $x\xrightarrow{i,j} x'$, where $K_i,x_i\in\mathbb{N}_0$ denote the number of processors and the queue-length within station $i\geq 0$. We finally have $K_1=1, K_2=\infty$ and $K_3=2$; at the virtual node, it holds $K_0 \wedge x_0 = 1$ always. Thus, note that transition rates in $X$ further depend on the network loads, and resemble kinetic laws within chemical reaction models \citep{golightly2015bayesian}.
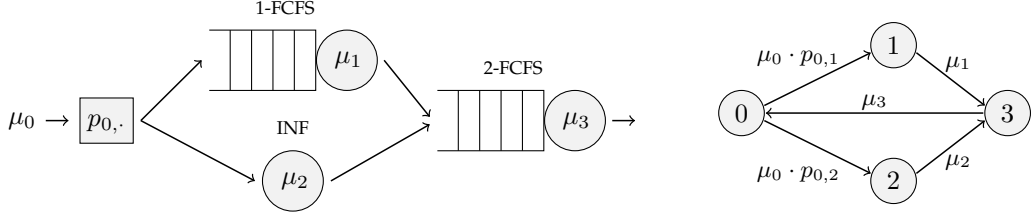
\begin{figure}[h!]
\centering
\begin{tikzpicture}
% Network 1
\draw (0,0.1) -- ++(1.4cm,0) -- ++(0,-0.8cm) -- ++(-1.4cm,0);
\foreach \i in {1,...,4}
  \draw (1.4cm-\i*8pt,0.1) -- +(0,-0.8cm);
\filldraw[fill=black!05!white] (1.4cm+0.41cm,-0.3cm) circle [radius=0.4cm]; % Circle
\node at (1,0.4cm) {\scriptsize 1-FCFS};

\filldraw[fill=black!05!white] (1.1cm,-1.4cm-0.5cm) circle [radius=0.4cm]; % Circle
\node at (1.1,-1.2cm) {\scriptsize INF};

\draw (3cm,-0.7cm) -- ++(1.4cm,0) -- ++(0,-0.8cm) -- ++(-1.4cm,0);
\foreach \i in {1,...,4}
  \draw (4.4cm-\i*8pt,-0.7cm) -- +(0,-0.8cm);
\filldraw[fill=black!05!white] (4.4cm+0.41cm,-0.7cm-0.4cm) circle [radius=0.4cm]; % Circle
\node at (4,-0.4cm) {\scriptsize 2-FCFS};

\filldraw[fill=black!05!white] (-1.7cm,-0.8cm) rectangle (-1cm,-1.4cm);

% the arrows and labels
\draw[->,line width=0.20mm] (-0.9cm,-1.1cm) -- (-0.1,-0.3cm);
\draw[->,line width=0.20mm] (-0.9cm,-1.1cm) -- (0.6,-1.9cm);
\draw[<-,line width=0.20mm] (-1.85cm,-1.1cm) -- +(-0.3cm,0cm) node[left] {$\mu_0$};
\draw[->,line width=0.20mm] (2.3cm,-0.3cm) -- (2.9cm,-1.05cm);
\draw[->,line width=0.20mm] (1.6cm,-1.9cm) -- (2.9cm,-1.15cm);
\draw[->,line width=0.20mm] (5.3cm,-1.1cm) -- +(0.3cm,0cm);

\node at (1.83,-0.3cm) {$\mu_1$};
\node at (1.13,-1.9cm) {$\mu_2$};
\node at (4.83,-1.1cm) {$\mu_3$};
\node[align=center] at (-1.35cm,-1.15cm) {$p_{0,\cdot}$};

% In terms of rates lambda
\filldraw[fill=black!05!white] (9cm,-0.1cm) circle [radius=0.3cm]; % Circle
\filldraw[fill=black!05!white] (10.5cm,-1cm) circle [radius=0.3cm]; % Circle
\filldraw[fill=black!05!white] (9cm,-1.9cm) circle [radius=0.3cm]; % Circle
\filldraw[fill=black!05!white] (7cm,-1cm) circle [radius=0.3cm]; % Circle
\node at (9,-0.1cm) {$1$};
\node at (10.5,-1) {$3$};
\node at (9,-1.9cm) {$2$};
\node at (7,-1cm) {$0$};
\draw[->,line width=0.20mm] (7.3cm,-0.9cm) -- (8.7,-0.2cm);
\draw[->,line width=0.20mm] (7.3cm,-1.1cm) -- (8.7,-1.8cm);
\draw[->,line width=0.20mm] (9.3cm,-0.2cm) -- (10.2,-0.9cm);
\draw[->,line width=0.20mm] (9.3cm,-1.8cm) -- (10.2,-1.1cm);
\draw[->,line width=0.20mm] (10.18cm,-1cm) -- (7.32,-1cm);
\node at (8.75,-0.85cm) {\small $\mu_{3}$};
\node at (9.85,-0.35cm) {\small $\mu_{1}$};
\node at (9.85,-1.65cm) {\small $\mu_{2}$};
\node at (7.75,-0.28cm) {\small $\mu_{0}\cdot p_{0,1}$};
\node at (7.75,-1.82cm) {\small $\mu_{0}\cdot p_{0,2}$};

\end{tikzpicture}
\vspace{10pt}
\caption{Left, open bottleneck network with $3$ service stations. Shaded circles indicate servers, queueing areas are pictured as empty rectangles. The box is a probabilistic junction for the routing of arrivals. Right, job transition intensities across network nodes.} \label{exmplQueues}
\end{figure}

\subsection{A hierarchical formulation of queueing systems}

Within a hierarchical multilevel formulation, rates in $\boldsymbol{\lambda}$ have a distribution (or \textit{image})  $\mathbb{P}_{\boldsymbol{\lambda}}\equiv \boldsymbol{\lambda}_*\mathbb{P}$ under a reference measure $\mathbb{P}$ on $(\Omega,\mathcal{F})$. We assume this to admit a density $f_{\boldsymbol{\lambda}}$ w.r.t. a base measure that will further induce (by properties of exponential transitions) distributions over the service rates $\boldsymbol{\mu}$ and routing topology. Next, note that a network trajectory over a finite interval is a piecewise deterministic jump process, such that $X\equiv (\boldsymbol{t},\boldsymbol{x})$ is represented by a sequence of transition times $\boldsymbol{t}$ along with states $\boldsymbol{x}$. Each pair $(\boldsymbol{t},\boldsymbol{x})$ is furthermore a random variable on a measurable space $(\mathcal{X},\Sigma_\mathcal{X})$ supporting finite $\mathcal{S}$-valued trajectories, and a conditional density $f_{X|\boldsymbol{\lambda}}$ may be defined w.r.t a dominating base measure $\mu_\mathcal{X}$, s.t. the regular conditional probability $\mathbb{P}(A|\boldsymbol{\lambda})$, $A\in\mathcal{F}$ satisfies
$$\mathbb{P}(X^{-1}(B)|\boldsymbol{\lambda}) = \int_B f_{X|\boldsymbol{\lambda}}(\boldsymbol{t},\boldsymbol{x}) \,\mu_\mathcal{X}(d\boldsymbol{t},d\boldsymbol{x})$$
for all $B\in\Sigma_{\mathcal{X}}$ (see Appendix \ref{notationApdx} for details). In this case, 
\begin{align}
f_{X|\boldsymbol{\lambda}}(\boldsymbol{t},\boldsymbol{x}) = \pi(x_0)\,
 e^{Q_{x_{I}}(T-t_I)}\, \prod_{i=1}^I Q_{x_{i-1},x_{i}}\, e^{Q_{x_{i-1}}(t_i-t_{i-1})}, \label{densPath}
\end{align}
for every pair of ordered times $\boldsymbol{t}=\{0,t_1,\dots,t_I\}$ in $[0,T]$ and states $\boldsymbol{x}=\{x_0,\dots,x_I\}$. Here, $\pi(\cdot)$ denotes an arbitrary distribution over initial states, and $Q\equiv Q(\boldsymbol{\lambda})$ is the matrix of infinitesimal rates associated with fixed values in $\boldsymbol{\lambda}$. The queueing network model is thus fully parametrized by a collection of hyper-parameters, and analogue modelling choices for \textit{continuous-time Markov chains} (CTMCs) or MJPs can be found in \cite{huelsenbeck2002inferring,baele2010using} or \cite{zhao2016bayesian}, to name a few. Finally, note that this set-up is not suitable for traditional probabilistic studies of queueing systems by means of balance equations, due to parameter uncertainty; however, we will show that it offers an appropriate framework for approximate transient analyses, parameter inference and reverse network evaluation (uncertainty quantification) tasks.

\subsection{Network evaluation and problem statement}

Let $T>0$ denote some arbitrary terminal time and $\boldsymbol{x}_0\in\mathcal{S}$ an initial state in $X$. For simplicity, this is assumed to be a $0$-valued vector, where no jobs populate the system. Now, let $0\leq t_1<\dots<t_K\leq T$ denote some fixed network monitoring times along with \textit{observation} variables $\{O_k\in\mathcal{O},k=1,\dots,K\}$, for some arbitrary support set $\mathcal{O}$, such that 
\begin{align}
\mathbb{P}\Big(\bigcap_{k=1}^K O_k^{-1}(\boldsymbol{o}_k)\big|X\Big) = \prod_{k=1}^K \mathbb{P}(O_k^{-1}(\boldsymbol{o}_k)|X)  = \prod_{k=1}^K f_{O|X_{t_k}}(\boldsymbol{o}_k) 
 \label{obsLik}
\end{align}
for any sequence of elements $\boldsymbol{o}_1,\dots,\boldsymbol{o}_K$ where $\boldsymbol{o}_k$ denotes the time-$t_k$ network observation across all nodes. Hence, any two observations are mutually independent if conditioned on their network states. The term $f_{O|X_{t_k}}$ stands for a conditional mass function assigned to measurements across the $M$ nodes; defined w.r.t a counting measure $\mu_{\mathcal{O}}$. In this paper it is assumed that $f_{O|x}>0$ (everywhere) for all $x\in\mathcal{S}$, however, deterministic observations such as queue lengths can be easily approximated by means of \textit{regularised} indicator functions; we discuss such examples within Section \ref{exmpls}. Extensions to continuous settings are straightforward.

Now, let $\mathbb{P}(A|\boldsymbol{o}_1,\dots,\boldsymbol{o}_K),\, A\in\mathcal{F},\, (\boldsymbol{o}_1,\dots,\boldsymbol{o}_K)\in  \mathcal{O}^K$ denote the regular conditional probability across global events and observations. Our interest lies in its induced distribution over the intensity rates (which we denote $\mathbb{P}_{\boldsymbol{\lambda}|\boldsymbol{o}_1,\dots,\boldsymbol{o}_K}$). Within (Bayesian) inferential settings, this induced distribution is referred to as a \textit{posterior}; it exists and admits a density carried by its corresponding \textit{prior} $\mathbb{P}_{\boldsymbol{\lambda}}$ (see Appendix \ref{notationApdx}), moreover, the transformation is proportional to a weighted product of network paths, and defined by the Radon-Nikodym derivative
\begin{align}
\frac{d\mathbb{P}_{\boldsymbol{\lambda}|\boldsymbol{o}_1,\dots,\boldsymbol{o}_K}}{d\mathbb{P}_{\boldsymbol{\lambda}}} = \frac{\int_{\mathcal{X}}\prod_{k=1}^K \mathbb{P}(O_k^{-1}(\boldsymbol{o}_k)|\boldsymbol{t},\boldsymbol{x})\, f_{X|\boldsymbol{\lambda}}(\boldsymbol{t},\boldsymbol{x})\, \mu_{\mathcal{X}}(d\boldsymbol{t},d\boldsymbol{x})}{\mathbb{P}(O_1^{-1}(\boldsymbol{o}_1)\cap\dots\cap O_k^{-1}(\boldsymbol{o}_k)) } , \label{likRatio}
\end{align}
which corresponds to Bayes' equation. There, the denominator denotes a normalising constant that integrates over trajectories and rates. This transformation will often induce a density representation $f_{\boldsymbol{\lambda}|\boldsymbol{o}_1,\dots,\boldsymbol{o}_K}$ for the posterior distribution w.r.t a suitable (Lebesgue) base. In these cases, we may think of the above derivative as a Likelihood-ratio. However, this ratio poses a tractability problem, that is, the integral over trajectories cannot be computed analytically and must be approximated. This is a common problem in inferential tasks with jump processes \citep[cf.][]{hobolth2009,rao13a,perez2019scalable}, and proposed solutions often rely on intensive MCMC procedures that iterate between trajectories and parameters; including direct sampling, rejection sampling or uniformization-based methods. Yet, algorithms are hard to implement, computationally demanding or only applicable to reduced classes of problems. In the case of queueing networks, strong temporal dependencies in the stochastic trajectories $X$ impose hard coupling properties amongst rates and paths \citep{sutton2011}, which limits the applicability of numerical state of the art solutions to the simplest types of network evaluation problems \citep{perez2017auxiliary}.

In the following, we present theoretical results leading to an alternative variational design to approximate the induced rate densities under the posterior measure in \eqref{likRatio}. For the purpose, we describe the inherent complexity of jump processes induced by networks of queues, and we further expose (and overcome) the multiple limitations of mean-field methods previously presented in \cite{opper2008variational} and references therein.

\section{Overview of results} \label{overSection}

Under the natural measure $\mathbb{P}$ tied to the infinitesimal generator $Q$, an underlying MJP $X$ as introduced in Section \ref{introSec} is supported in a set $\mathcal{S}$ of feasible vectors of integers, which is often just $\mathcal{S}=\mathbb{N}_0^{|\mathcal{C}|\times M}$.  Now, assume the existence of an approximating measure $\tilde{\mathbb{P}}$  on an augmented space of network paths $\tilde{\mathcal{X}}$, such that we further assign a mass to network loads with negative values. Rates for transitions across the states are induced by a generator $\tilde{Q}$ with 
\begin{align}
\tilde{Q}_{x,x'}=\delta+Q_{x,x'}\, , \quad \delta>0\, \label{deltaGenerator}
\end{align}
whenever $x\xrightarrow{i,j,c} x'$ is such that $(i,j,c)\in\mathcal{T}$, and $\tilde{Q}_{x,x'}=Q_{x,x'}=0$ otherwise. Hence, intensities for job transitions between nodes $i$ and $j$ are strictly positive whenever $p^c_{i,j}>0$, for any class and regardless of the network loads. In the event of a $t$-time class-$c$ job departure from a station $i$ when $X^{i,c}_t\leq0$, then we assume this job to be virtually generated and a unit will be subtracted from the state vector at the corresponding index, in order to represent the fact and preserve the global population count. For values of $\delta$ small enough, the $\tilde{\mathbb{P}}$-density assigned to trajectories outside of $\mathcal{X}$ is negligible. Note that a density $\tilde{f}$ in \eqref{densPath} with generator $\tilde{Q}$ in \eqref{deltaGenerator} is such that, for any network path $(\boldsymbol{t},\boldsymbol{x})\in\tilde{\mathcal{X}}\backslash\mathcal{X}$, it holds $$\tilde{f}_{X|\boldsymbol{\lambda}}(\boldsymbol{t},\boldsymbol{x})\leq \prod_{i=1}^I \tilde{Q}_{x_{i-1},x_i}=O(\delta^r)\quad \text{as} \quad\delta\rightarrow 0$$ for some $r\in\{1,\dots,I\}$. Thus, $X_*\tilde{\mathbb{P}}(\mathcal{X})=1-\int_{\tilde{\mathcal{X}}\backslash\mathcal{X}} \tilde{f}_{X|\boldsymbol{\lambda}}d\tilde{\mu}_{\mathcal{X}}\xrightarrow{\delta\rightarrow 0} 1$, where  $\tilde{\mu}$ denotes an appropriately augmented base measure, and the limiting system dynamics under $\tilde{\mathbb{P}}$ will offer a perfect approximation to the original network model. Within the rest of the paper,
\begin{itemize}
\item In Section \ref{countModelSec}, we present a counting process over job transitions in the augmented network with generator $\tilde{Q}$ in \eqref{deltaGenerator}, and introduce an alternative absolutely continuous mean-field measure $\mathbb{Q}$. In Lemma \ref{lowerBoundLemma}, we derive a lower bound to the equivalent log-likelihood for the network measurements.
\item Propositions \ref{Prop1} and \ref{Prop2} within Section \ref{FunctionalSec} inspect the structure of $\mathbb{Q}$ that approximates the regular conditional probability and corresponding likelihood-ratio in \eqref{likRatio}. Corollaries \ref{Cor1} and \ref{Cor2} focus on the rate density $d\mathbb{P}_{\boldsymbol{\lambda}|\boldsymbol{o}_1,\dots,\boldsymbol{o}_K}$ by looking at conjugacy properties and limiting behaviour as $\delta\rightarrow 0$.
\item Finally, Section \ref{exmpls} describes applications of our results within inferential and evaluation tasks, allowing to approximate (image) measures across the various service rates $\mu$ and routing probabilities in $\mathcal{P}$, conditioned on network measurements. This includes comparisons with existing alternative methods.
\end{itemize}

\section{A counting process over job transitions} \label{countModelSec}

A network system as introduced in Section \ref{introSec} further gives raise to a multivariate Markov counting process $Y=(Y_t)_{t\geq 0}$ on $(\Omega,\mathcal{F})$, where each indexed $Y_t=(Y^{\boldsymbol{\eta}}_t)_{{\boldsymbol{\eta}}\in\mathcal{T}}\in\mathcal{S}^Y$ accounts for job transitions across all classes in $\mathcal{T}$, up to a time $t\geq 0$. That is, each $Y_t^{{\boldsymbol{\eta}}}$ denotes the cumulative count in $Y$ of transitions $x\xrightarrow{{\boldsymbol{\eta}}} x'$ in $X$, with $x,x'\in\mathcal{S}$, and $Y_0^{{\boldsymbol{\eta}}}=0$ for all $\boldsymbol{\eta}\in\mathcal{T}$. At a basic level, these are simply non-decreasing counting processes for job transitions in the directions defined within $\mathcal{T}$. We further note that $|\mathcal{T}|$ is often small, as underlying network topologies impose strict routing mechanisms. The support set $\mathcal{S}^Y$ for the counting process is determined by the connectivity structure amongst the stations. Under the approximating measure $\tilde{\mathbb{P}}$, it holds $\mathcal{S}^Y=\mathbb{N}_0^{|\mathcal{T}|}$, since job transitions may occur regardless of network loads. Now, let $$\mathcal{T}_{i,c}^{\leftarrow}=\{\boldsymbol{{\boldsymbol{\eta}}}\in\mathcal{T} :\eta_2=i,\eta_3=c\} \quad \text{and} \quad \mathcal{T}_{i,c}^{\rightarrow}=\{\boldsymbol{\eta}\in\mathcal{T} :\eta_1=i,\eta_3=c\}$$
denote the subsets of $\mathcal{T}$ that include class $c\in\mathcal{C}$ job transitions to, and from, the network node $i\in\{0,\dots,M\}$, respectively. Also, recall that $X_t^{i,c}$ denotes the number of class $c$ jobs in station $i>0$ at time $t\geq 0$, then 
\begin{align}
X^{i,c}_t = \sum_{{\boldsymbol{\eta}}\in\mathcal{T}^{\leftarrow}_{i,c}}  Y_t^{{\boldsymbol{\eta}}} -  \sum_{{\boldsymbol{\eta}}\in\mathcal{T}^{\rightarrow}_{i,c}} Y_t^{{\boldsymbol{\eta}}}   \label{relXandY}
\end{align}
for all $t\geq 0$, assuming initially empty networked systems. We note that for all ${\boldsymbol{\eta}}=(i,j,c)\in\mathcal{T}$ it holds ${\boldsymbol{\eta}}\in\mathcal{T}_{j}^{\leftarrow}$ and ${\boldsymbol{\eta}}\in\mathcal{T}_{i}^{\rightarrow}$. Thus, paths in $X$ and $Y$ differ in that the former is \textit{coupled}, i.e. a job transition in the direction $\boldsymbol{\eta}=(i,j,c)$ is relevant to (and thus is synchronized across) a pair of marginal processes $(X^{i,c}_t)_{t\geq0},(X^{j,c}_t)_{t\geq0}$; in the latter, this is only relevant to the indexed process $(Y^{\boldsymbol{\eta}}_t)_{t\geq 0}$.

In view of \eqref{relXandY}, we further denote $x_{i,c} = \sum_{{\boldsymbol{\eta}}\in\mathcal{T}^{\leftarrow}_{i,c}}  y_{{\boldsymbol{\eta}}} -  \sum_{{\boldsymbol{\eta}}\in\mathcal{T}^{\rightarrow}_{i,c}} y_{{\boldsymbol{\eta}}} $ to the class-$c$ queue-length in station $i>0$ for any $y\in\mathcal{S}^Y$. Then, the $\tilde{\mathbb{P}}$-associated infinitesimal generator matrix $\Xi$ of $Y$ is such that $\textstyle\Xi_{y,y'}\equiv\textstyle\Xi_{y,{\boldsymbol{\eta}}} =\delta + \lambda_{{\boldsymbol{\eta}}}\cdot\big[\Upsilon(y,\eta_1,\eta_3)\vee 0\big]$ with a station load
\begin{align}
\textstyle\Upsilon(y,i,c) = x_{i,c} \cdot \Bigg(\dfrac{K_i}{\sum_{c'\in\mathcal{C}}x_{i,c'}} \wedge 1 \Bigg) \label{Upsil1}
\end{align}
for all jumps $y\xrightarrow{{\boldsymbol{\eta}}} y'$, ${\boldsymbol{\eta}}=(i,j,c)$, where the origin station $i>0$ has PS discipline (here we have set $0/0=0$), and
\begin{align}
\textstyle\textstyle\Upsilon(y,i,c) = K_i \wedge x_{i,c}  \label{Upsil2}
\end{align}
in stations $i>0$ with FCFS policy within single-class networks. We further have $\textstyle\Xi_{y,y'}=\delta+\lambda_{0,j,c}$ for arrivals from virtual nodes (in open networks) and $\textstyle\Xi_{y,y'}=\delta+\lambda_{0,j,c} \cdot (N+\sum_{{\boldsymbol{\eta}}\in\mathcal{T}^{\leftarrow}_{0,c}}  y_{{\boldsymbol{\eta}}} -  \sum_{{\boldsymbol{\eta}}\in\mathcal{T}^{\rightarrow}_{0,c}} y_{{\boldsymbol{\eta}}})$ for arrivals from delays, where $N$ denotes the job population in a closed system. Finally, $\Xi_y \coloneqq \Xi_{y,y} = - \sum_{y'\in\mathcal{S}^Y: y\neq y'} \Xi_{y,y'}$.

\subsection{A mean field decomposition and lower bound}

The likelihood for observation events in \eqref{obsLik} readily transfers to counts $Y$ by means of \eqref{relXandY}, we thus may write $ f_{O|Y_{t_k}}(\boldsymbol{o}_k)\equiv f_{O|X_{t_k}}(\boldsymbol{o}_k) $. Under the measure $\tilde{\mathbb{P}}$ network states can have negative values, the likelihood is undefined in such instances. Now, note that piecewise $\mathcal{S}^Y$-valued trajectories also represent elements $(\boldsymbol{t},\boldsymbol{y})$ in a space $(\mathcal{Y},\Sigma_{\mathcal{Y}})$, similar to network paths in $X$. Let $f_{Y|\boldsymbol{\lambda},\boldsymbol{o}_1,\dots,\boldsymbol{o}_K}$ be a density function, w.r.t. some base measure $\mu_\mathcal{Y}$, where for all $B\in\Sigma_{\mathcal{Y}}$ it holds $$\mathbb{P}(Y^{-1}(B)|\boldsymbol{\lambda},\boldsymbol{o}_1,\dots,\boldsymbol{o}_K) = \int_B f_{Y|\boldsymbol{\lambda},\boldsymbol{o}_1,\dots,\boldsymbol{o}_K} \,d\mu_\mathcal{Y}.$$ It may be shown by properties of conditional distributions that, conditioned on observations, $Y$ is a non-homogeneous semi-Markov process with \textit{hazard} functions
\begin{equation}
\Lambda_{y,y'}(t)=\Xi_{y,y'}\cdot\frac{\mathbb{P}(\bigcap_{k : t_k>t}O_k^{-1}(\boldsymbol{o_k})|Y_t=y')}{\mathbb{P}(\bigcap_{k : t_k>t}O_k^{-1}(\boldsymbol{o_k})|Y_t=y)} \label{hazards}
\end{equation}
for $y'\neq y$, and $\Lambda_{y}(t)=-\sum_{y'\neq y}\Lambda_{y,y'}(t)$,  s.t.
\begin{align*}
f_{Y|\boldsymbol{\lambda},\boldsymbol{o}_1,\dots,\boldsymbol{o}_K}(\boldsymbol{t},\boldsymbol{y}) =  \pi(y_0)\, e^{\int_{t_I}^{T} \Lambda_{y_{I}}(u)du} \, \prod_{i=1}^I \Lambda_{y_{i-1},y_{i}}(t_i)\, e^{\int_{t_{i-1}}^{t_i} \Lambda_{y_{i-1}}(u)du}\, .
\end{align*}
Here, $\Xi\equiv \Xi(\boldsymbol{\lambda})$ denotes the generator matrix associated with fixed values in $\boldsymbol{\lambda}$. For a deeper look at conditional jump processes we refer the reader to \cite{serfozo1972conditional,daley2007introduction}. This conditional counting process is of key importance, however, the structure of rates in \eqref{hazards} poses a trivial analytical impediment. In our approximating effort, we assume the existence of an alternative measure $\mathbb{Q}$ on $(\Omega,\mathcal{F})$. Under this measure, network trajectories in $X$ are subject to a mean-field decomposition across interacting pairwise nodes, that is, the $\mathbb{Q}$-law of $Y$ is that of a family of $|\mathcal{T}|$ independent non-homogeneous Poisson counting processes with state-dependent intensity functions $\nu^{\boldsymbol{\eta}} = (\nu^{\boldsymbol{\eta}}_t(\cdot))_{t\geq 0}$, for all ${\boldsymbol{\eta}}\in\mathcal{T}$. Here,
\begin{itemize}
\item Intensity rates for jumps $y\xrightarrow{t,{\boldsymbol{\eta}}} y'$, $t\geq 0$, are independent of $\boldsymbol{\lambda}$, change over time, and are given by $\nu^{\boldsymbol{\eta}}_t(y_{\boldsymbol{\eta}})$. 
\item Holding rates in $Y$ evolve according to $|\nu_t (Y_t)|$, with $\nu_t (Y_t)=-\sum_{\boldsymbol{\eta}\in\mathcal{T}}\nu^{\boldsymbol{\eta}}_t(Y^{\boldsymbol{\eta}}_t)$.
\item The state probability of the multivariate process $Y$ factors across the job transition directions, s.t.
$$\mathbb{Q}(Y_t=y) = \prod_{\boldsymbol{\eta}\in\mathcal{T}}\mathbb{Q}(Y^{\boldsymbol{\eta}}_t=y_{\boldsymbol{\eta}})$$ for every $y\in\mathcal{S}^Y$. 
\end{itemize}
In order to ensure computational tractability within forthcoming procedures, the intensity functions $\nu$ must be bounded from above by some arbitrary constant, s.t. $\nu^{\boldsymbol{\eta}}_t(y_{\boldsymbol{\eta}})\leq \bar{\nu}$ for all $t>0,\boldsymbol{\eta}\in\mathcal{T}$ and $y\in\mathcal{S}^Y$. Furthermore, transition rates in $\boldsymbol{\lambda}$ are assumed mutually independent under $\mathbb{Q}$, and admit undetermined densities $d\mathbb{Q}_{\lambda,\boldsymbol{\eta}}$, $\boldsymbol{\eta}\in\mathcal{T}$, that must integrate to $1$ on $(\mathbb{R}_+,\mathcal{B}(\mathbb{R}_+))$. We will later observe that this still induces dependence structures across the services rates and routing probabilities in the original network model. We finally note that $\mathbb{Q}$ and $\tilde{\mathbb{P}}$ are equivalent on $\mathcal{F}$, as both assign a positive measure to every marginally-increasing sequence of $\mathbb{N}_0^{|\mathcal{T}|}$-valued counts. 

\begin{lemm}[Mean-field lower bound] \label{lowerBoundLemma} Denote $\boldsymbol{O}=\bigcap_{k=1}^K O_k^{-1}(\boldsymbol{o}_k)$  and let $\tilde{\mathbb{P}}$ and $\mathbb{Q}$ be the probability measures on $(\Omega,\mathcal{F})$, as defined above. Recall notation $\textstyle\Xi_{y,y'}\equiv\textstyle\Xi_{y,{\boldsymbol{\eta}}}$ for jumps $y\xrightarrow{{\boldsymbol{\eta}}} y'$ with direction ${\boldsymbol{\eta}}$, then 
\begin{align}
\log\tilde{\mathbb{P}}(\boldsymbol{O}) &\geq \sum_{k=1}^K\mathbb{E}^{\mathbb{Q}}_{Y_{t_k}} \big[ \log  f_{O|Y_{t_k}}(\boldsymbol{o}_k)  \big]
 -\mathbb{E}^{\mathbb{Q}}_{\boldsymbol{\lambda}} \bigg[  \log\frac{ d\mathbb{Q}_{\boldsymbol{\lambda}} }{d\mathbb{P}_{\boldsymbol{\lambda}}} \bigg] \nonumber \\ 
&- \int_0^T  \mathbb{E}^{\mathbb{Q}}_{Y_t,\boldsymbol{\lambda}} \bigg[  \sum_{{\boldsymbol{\eta}}\in\mathcal{T}}  \nu^{\boldsymbol{\eta}}_t (Y^{\boldsymbol{\eta}}_t) \log \frac{ \nu^{\boldsymbol{\eta}}_t(Y^{\boldsymbol{\eta}}_t)}{\Xi_{Y_t,{\boldsymbol{\eta}}}}  - \Xi_{Y_t} + \nu_t(Y_t) \bigg] dt \label{varboundDeveloped}
\end{align}
offers a lower bound on the $\tilde{\mathbb{P}}$-probability of retrieved observation events.
\end{lemm}
\begin{proof} Note that
\begin{align}
\log\tilde{\mathbb{P}}(\boldsymbol{O})&=\log\int_{\mathcal{Y}\times\mathbb{R}_+^{|\mathcal{T}|}} \mathbb{P}(\boldsymbol{O}|Y)\, d(Y,\boldsymbol{\lambda})_*\tilde{\mathbb{P}} = \log \mathbb{E}^{\mathbb{Q}}_{Y,\boldsymbol{\lambda}} \bigg[ \mathbb{P}(\boldsymbol{O}|Y)\, \frac{d(Y,\boldsymbol{\lambda})_*\tilde{\mathbb{P}}}{d(Y,\boldsymbol{\lambda})_*\mathbb{Q}} \bigg] \nonumber \\
&\geq \mathbb{E}^{\mathbb{Q}}_{Y} \big[ \log  \mathbb{P}(\boldsymbol{O}|Y)\big]  - \mathbb{E}^{\mathbb{Q}}_{Y,\boldsymbol{\lambda}} \bigg[ \log  \frac{d(Y,\boldsymbol{\lambda})_*\mathbb{Q}}{d(Y,\boldsymbol{\lambda})_*\tilde{\mathbb{P}}} \bigg] \label{varbound}
\end{align}
where we use Jensen's inequality for finite measures. This is known as a variational mean-field lower bound on the log-likelihood, and
\begin{align*}
\mathbb{E}^{\mathbb{Q}}_{Y} \big[ \log  \mathbb{P}(\boldsymbol{O}|Y)\big] &= \mathbb{E}^{\mathbb{Q}}_{Y} \big[ \log   \prod_{k=1}^K f_{O|Y_{t_k}}(\boldsymbol{o}_k)  \big] = \sum_{k=1}^K\mathbb{E}^{\mathbb{Q}}_{Y_{t_k}} \big[ \log  f_{O|Y_{t_k}}(\boldsymbol{o}_k)  \big] 
\end{align*}
follows directly from \eqref{obsLik}. The negative part in \eqref{varbound} is the \textit{Kullback-Leibler} (KL) divergence between image measures of  $\mathbb{Q}$ and $\tilde{\mathbb{P}}$. By noting that these share base measures, and $Y,\boldsymbol{\lambda}$ are independent under $\mathbb{Q}$, it holds
\begin{align}
\mathbb{E}^{\mathbb{Q}}_{Y,\boldsymbol{\lambda}} \bigg[ \log  \frac{d(Y,\boldsymbol{\lambda})_*\mathbb{Q}}{d(Y,\boldsymbol{\lambda})_*\tilde{\mathbb{P}}} \bigg] &= \mathbb{E}^{\mathbb{Q}}_{\boldsymbol{\lambda}} \bigg[  \log\frac{ d\mathbb{Q}_{\boldsymbol{\lambda}} }{d\mathbb{P}_{\boldsymbol{\lambda}}} \bigg]
 \, + \, \mathbb{E}^{\mathbb{Q}}_{\boldsymbol{\lambda}} \bigg[ \mathbb{E}^{\mathbb{Q}}_{Y} \bigg[ \log  \frac{g_Y}{f_{Y|\boldsymbol{\lambda}}}  \bigg] \bigg] , \label{eq1Proof1}
\end{align}
where $g_Y$ and $f_{Y|\boldsymbol{\lambda}}$ denote the $Y$-trajectory densities associated with rates $\nu^{\boldsymbol{\eta}}$ and $\Xi(\boldsymbol{\lambda})$, respectively. The last term in \eqref{eq1Proof1} is a $\mathbb{Q}$-average of the KL divergence on $Y$,  where the mean is taken across the infinitesimal transition rates. For a fixed starting $Y_0\in\mathcal{S}^Y$, the inner expectation is shown in \cite{opper2008variational} to take the equivalent form 
\begin{align*}
\mathbb{E}^{\mathbb{Q}}_{Y} \bigg[ \log  \frac{g_Y}{f_{Y|\boldsymbol{\lambda}}} \bigg] &= \lim_{R\rightarrow\infty} \sum_{r=0}^{R-1}  \mathbb{E}^{\mathbb{Q}}_{Y_{\frac{Tr}{R}}} \Bigg[ \sum_{y\in\mathcal{S}^Y} \mathbb{Q}(Y_{\frac{T(r+1)}{R}}=y|Y_{\frac{Tr}{R}}) \log\frac{\mathbb{Q}(Y_{\frac{T(r+1)}{R}}=y|Y_{\frac{Tr}{R}})}{\tilde{\mathbb{P}}(Y_{\frac{T(r+1)}{R}}=y|Y_{\frac{Tr}{R}},\boldsymbol{\lambda})} \Bigg] .
\end{align*}
Note that within an infinitesimal time interval a jump in $Y$ may only happen in one direction within $\mathcal{T}$. With this in mind, we retrieve the limit of a Riemann sum in the interval $[0,T]$, i.e.
\begin{align*}
\mathbb{E}^{\mathbb{Q}}_{Y} \bigg[ \log  \frac{g_Y}{f_{Y|\boldsymbol{\lambda}}} \bigg] & = \lim_{R\rightarrow\infty} \frac{T}{R }\sum_{r=0}^{R-1}  \mathbb{E}^{\mathbb{Q}}_{Y_{\frac{Tr}{R}}} \Bigg[  \sum_{{\boldsymbol{\eta}}\in\mathcal{T}}  \nu^{\boldsymbol{\eta}}_{\frac{Tr}{R}}(Y^{\boldsymbol{\eta}}_{\frac{Tr}{R}}) \log \frac{ \nu^{\boldsymbol{\eta}}_{\frac{Tr}{R}}(Y^{\boldsymbol{\eta}}_{\frac{Tr}{R}})}{\Xi_{Y_{\frac{Tr}{R}},{\boldsymbol{\eta}}}}  +  \frac{R}{T} \log\frac{1 + \frac{T}{R}  \nu_{\frac{Tr}{R}}(Y_{\frac{Tr}{R}})}{1+ \frac{T}{R}   \Xi_{Y_{\frac{Tr}{R}}}}
 \Bigg] \\
 & = \int_0^T  \mathbb{E}^{\mathbb{Q}}_{Y_t} \bigg[  \sum_{{\boldsymbol{\eta}}\in\mathcal{T}}  \nu^{\boldsymbol{\eta}}_t (Y^{\boldsymbol{\eta}}_t) \log \frac{ \nu^{\boldsymbol{\eta}}_t(Y^{\boldsymbol{\eta}}_t)}{\Xi_{Y_t,{\boldsymbol{\eta}}}}  - \Xi_{Y_t} + \nu_t(Y_t) \bigg] dt ,
\end{align*}
and
\begin{align*}
\mathbb{E}^{\mathbb{Q}}_{Y,\boldsymbol{\lambda}} \bigg[ \log  \frac{d(Y,\boldsymbol{\lambda})_*\mathbb{Q}}{d(Y,\boldsymbol{\lambda})_*\tilde{\mathbb{P}}} \bigg] = \mathbb{E}^{\mathbb{Q}}_{\boldsymbol{\lambda}} \bigg[  \log\frac{ d\mathbb{Q}_{\boldsymbol{\lambda}} }{d\mathbb{P}_{\boldsymbol{\lambda}}} \bigg] + \int_0^T  \mathbb{E}^{\mathbb{Q}}_{Y_t,\boldsymbol{\lambda}} \bigg[  \sum_{{\boldsymbol{\eta}}\in\mathcal{T}}  \nu^{\boldsymbol{\eta}}_t (Y^{\boldsymbol{\eta}}_t) \log \frac{ \nu^{\boldsymbol{\eta}}_t(Y^{\boldsymbol{\eta}}_t)}{\tilde{\Xi}_{Y_t,{\boldsymbol{\eta}}}}  - \tilde{\Xi}_{Y_t} + \nu_t(Y_t) \bigg] dt
\end{align*}
completes the proof.
\end{proof}

Thus, the lower bound in \eqref{varboundDeveloped} depends on both the \textit{latent} variables $Y$ and $\boldsymbol{\lambda}$, accounting for the various counts and rates. On a basic level, this is built by $3$ distinguishable components, that is, the expected log-observations, the \textit{Kullback-Leibler} divergence across service rate densities, and a $\mathbb{Q}$-weighted divergence across hazard functions and rates, further integrated along the entire network trajectory.

\section{A functional representation} \label{FunctionalSec}

The above bound includes the prior rates density $d\mathbb{P}_{\boldsymbol{\lambda}}$ along with the $\tilde{\mathbb{P}}$-generator $\Xi$ for the approximating network system with negative loads. In addition, we can find the unknown $\mathbb{Q}$ distribution for the time-indexed random variables $Y_t$, along with undetermined hazard rates $\nu$ and densities for infinitesimal rates in $\boldsymbol{\lambda}$. Hence, by maximising this bound, we may derive properties on $\mathbb{Q}$ that allow for the construction of an approximating distribution to $d\mathbb{P}_{\boldsymbol{\lambda}|\boldsymbol{o}_1,\dots,\boldsymbol{o}_K}$ and the corresponding likelihood-ratio in \eqref{likRatio}. In this Section, we begin by generalizing work in \cite{opper2008variational} and present results that (i) accommodate parameter uncertainty in transition rates and (ii) impose computational restrictions in the resulting iterative system of equations. Later, we move on to inspect posterior rate densities and conjugacy properties as $\delta\rightarrow 0$.
 
\begin{prop}  \label{Prop1}
Let $d\mathbb{Q}_{\boldsymbol{\lambda}}$ be some valid joint density assigned to the instantaneous rates $\boldsymbol{\lambda}$ under the approximating mean-field measure $\mathbb{Q}$. Also, define $Y^{\backslash\boldsymbol{\eta}}_t = \{ Y^{\boldsymbol{\eta}'}_t  \, :\, \boldsymbol{\eta}'\in\mathcal{T} \backslash \{\boldsymbol{\eta}\}\}$. Then, the $\mathbb{Q}$-dynamics of $Y$ that optimize the lower bound  \eqref{varboundDeveloped} may be parametrized by a system of equations, so that the intensity functions $\nu_t^{\boldsymbol{\eta}}(y)\leq\bar{\nu}$ are given by  
\begin{align}
\nu_t^{\boldsymbol{\eta}}(y) =  \frac{r_t^{\boldsymbol{\eta}}(y+1)}{r_t^{\boldsymbol{\eta}}(y)} e^{\mathbb{E}^{\mathbb{Q}}_{Y^{\backslash\boldsymbol{\eta}}_t,\boldsymbol{\lambda}} \big[ \log\Xi_{Y_t,{\boldsymbol{\eta}}} | Y_t^{\boldsymbol{\eta}}=y \big] - k_t^{\boldsymbol{\eta}}(y)/\mathbb{Q}(Y_t^{\boldsymbol{\eta}}=y) }  \label{ratesEq}
\end{align}
for all $t\in[0,T]$, $\boldsymbol{\eta}\in\mathcal{T}$ and $y\in\mathbb{N}_0$, with $\kappa_t^{\boldsymbol{\eta}}(y)\geq 0$ and
\begin{align}
\frac{dr_t^{\boldsymbol{\eta}}(y)}{dt} = r_t^{\boldsymbol{\eta}}(y)  \mathbb{E}^{\mathbb{Q}}_{Y^{\backslash\boldsymbol{\eta}}_t,\boldsymbol{\lambda}} \big[ \Xi_{Y_t,{\boldsymbol{\eta}}} | Y_t^{\boldsymbol{\eta}}=y \big] - \bigg(1+\frac{k_t^{\boldsymbol{\eta}}(y)}{\mathbb{Q}(Y_t^{\boldsymbol{\eta}}=y)}\bigg)  r_t^{\boldsymbol{\eta}}(y+1) \frac{e^{\mathbb{E}^{\mathbb{Q}}_{Y^{\backslash\boldsymbol{\eta}}_t,\boldsymbol{\lambda}} \big[ \log\Xi_{Y_t,{\boldsymbol{\eta}}} | Y_t^{\boldsymbol{\eta}}=y \big]}}{e^{ k_t^{\boldsymbol{\eta}}(y)/\mathbb{Q}(Y_t^{\boldsymbol{\eta}}=y)}} \label{diffR}
\end{align}
whenever $t\neq t_k$, $k=1,\dots,K$, and 
\begin{align}
\lim_{t\rightarrow t^-_k} r_t^{\boldsymbol{\eta}}(y) = r_{t_k}^{\boldsymbol{\eta}}(y) \exp ( \mathbb{E}^{\mathbb{Q}}_{Y^{\backslash\boldsymbol{\eta}}_{t_k}} \big[ \log  f_{O|Y_{t_k}}(\boldsymbol{o}_k) | Y_{t_k}^{\boldsymbol{\eta}}=y \big] ) \label{jumpR}
\end{align} 
at network observation times. In addition, $\kappa_t^{\boldsymbol{\eta}}(y)(\nu_t^{\boldsymbol{\eta}}(y)-\bar{\nu})=0$.
\end{prop}

\begin{proof}
We identify a stationary point to the Lagrangian associated with this constrained optimization problem, where optimization is w.r.t. the jump rates and the finite dimensional distributions of $Y$, subject to $\nu_t^{\boldsymbol{\eta}}(y)\leq \bar{\nu}$ and the master equation
\begin{align}
\frac{d\mathbb{Q}(Y_t^{\boldsymbol{\eta}}=y)}{dt} =\nu_t^{\boldsymbol{\eta}}(y-1)\cdot\mathbb{Q}(Y_t^{\boldsymbol{\eta}}=y-1) - \nu_t^{\boldsymbol{\eta}}(y)\cdot \mathbb{Q}(Y_t^{\boldsymbol{\eta}}=y) \label{masterEq}
\end{align}
for $y\geq 1$, with $d\mathbb{Q}(Y_t^{\boldsymbol{\eta}}=0)= - \nu_t^{\boldsymbol{\eta}}(0) \mathbb{Q}(Y_t^{\boldsymbol{\eta}}=0)dt$. Denote by $\phi_t^{\boldsymbol{\eta}}(y) = \mathbb{Q}(Y_t^{\boldsymbol{\eta}}=y)$ the functional representing the marginal $\mathbb{Q}$-probability of the state $Y_t$ in the direction of $\boldsymbol{\eta}$, for all $y\in\mathbb{N}_0$. In view of \eqref{varboundDeveloped}, the object function may be expressed as the functional
\begin{align*}
\Phi[\phi,\nu,l]&= C + \sum_{k=1}^K \mathbb{E}^{\mathbb{Q}}_{Y_{t_k}} \big[ \log  f_{O|Y_{t_k}}(\boldsymbol{o}_k) \big] -  \int_0^T  \sum_{{\boldsymbol{\eta}}\in\mathcal{T}}  \mathbb{E}^{\mathbb{Q}}_{Y^{\boldsymbol{\eta}}_t} \big[ \Psi[Y^{\boldsymbol{\eta}}_t,\phi_t^{\boldsymbol{\eta}}(Y_t^{\boldsymbol{\eta}}) , \nu_t^{\boldsymbol{\eta}}(Y_t^{\boldsymbol{\eta}}),l_t^{\boldsymbol{\eta}}(Y_t^{\boldsymbol{\eta}})] \big] dt 
\end{align*}
with
\begin{align*}
\Psi[Y^{\boldsymbol{\eta}}_t,\phi_t^{\boldsymbol{\eta}}(Y_t^{\boldsymbol{\eta}}) , \nu_t^{\boldsymbol{\eta}}(Y_t^{\boldsymbol{\eta}}),l_t^{\boldsymbol{\eta}}(Y_t^{\boldsymbol{\eta}})] &= \nu_t^{\boldsymbol{\eta}}(Y_t^{\boldsymbol{\eta}}) \Big( \log  \nu_t^{\boldsymbol{\eta}}(Y_t^{\boldsymbol{\eta}}) -  \mathbb{E}^{\mathbb{Q}}_{Y^{\backslash\boldsymbol{\eta}}_t,\boldsymbol{\lambda}} \big[ \log \Xi_{Y_t,{\boldsymbol{\eta}}} |  Y^{\boldsymbol{\eta}}_t  \big]  - 1\Big) \\
& + \mathbb{E}^{\mathbb{Q}}_{Y^{\backslash\boldsymbol{\eta}}_t,\boldsymbol{\lambda}} \big[ \Xi_{Y_t,{\boldsymbol{\eta}}}  |  Y^{\boldsymbol{\eta}}_t \big]  - l_t^{\boldsymbol{\eta}}(Y_t^{\boldsymbol{\eta}}) \Big(\nu_t^{\boldsymbol{\eta}}(Y_t^{\boldsymbol{\eta}}) + \frac{d\log \phi_t^{\boldsymbol{\eta}}(Y_t^{\boldsymbol{\eta}})}{dt} \Big) \\
& + l_t^{\boldsymbol{\eta}}(Y_t^{\boldsymbol{\eta}}+1)\nu_t^{\boldsymbol{\eta}}(Y_t^{\boldsymbol{\eta}}) - \frac{k_t^{\boldsymbol{\eta}}(Y_t^{\boldsymbol{\eta}})}{\phi_t^{\boldsymbol{\eta}}(Y_t^{\boldsymbol{\eta}})} \big(\bar{\nu}-\nu_t^{\boldsymbol{\eta}}(Y_t^{\boldsymbol{\eta}})\big)   ,
\end{align*}
where $l^{\boldsymbol{\eta}} = (l^{\boldsymbol{\eta}}_t(\cdot))_{t\geq 0}$ and $k^{\boldsymbol{\eta}} = (k^{\boldsymbol{\eta}}_t(\cdot))_{t\geq 0}$ are multiplier functions that ensure \eqref{masterEq} and the complementary inequality on rates are satisfied. Above, the term $C$ includes the remainder bits in the lower bound in \eqref{varboundDeveloped} that are independent of the finite dimensional distributions of $Y$ under $\mathbb{Q}$. Hence, we obtain the following functional derivatives
\begin{align*}
\fdv{\Phi}{\phi_t^{\boldsymbol{\eta}}(y)} &= \sum_{k=1}^K \delta(t-t_k) \,
\mathbb{E}^{\mathbb{Q}}_{Y^{\backslash\boldsymbol{\eta}}_t} \big[ \log  f_{O|Y_t}(\boldsymbol{o}_k) | Y_t^{\boldsymbol{\eta}}=y \big]  - \mathbb{E}^{\mathbb{Q}}_{Y^{\backslash\boldsymbol{\eta}}_t,\boldsymbol{\lambda}} \big[ \Xi_{Y_t,{\boldsymbol{\eta}}} | Y_t^{\boldsymbol{\eta}}=y \big] - \frac{d l_t^{\boldsymbol{\eta}}(y)}{dt}  \\
& - \nu_t^{\boldsymbol{\eta}}(y) \Big( \log  \nu_t^{\boldsymbol{\eta}}(y) -  \mathbb{E}^{\mathbb{Q}}_{Y^{\backslash\boldsymbol{\eta}}_t,\boldsymbol{\lambda}} \big[ \log \Xi_{Y_t,{\boldsymbol{\eta}}} |  Y^{\boldsymbol{\eta}}_t = y \big]  - 1 -  l_t^{\boldsymbol{\eta}}(y) +  l_t^{\boldsymbol{\eta}}(y+1) \Big) \ ,
\end{align*}
and
\begin{align*}
\fdv{\Phi}{\nu_t^{\boldsymbol{\eta}}(y)} &= - \phi_t^{\boldsymbol{\eta}}(y)
 \big( \log  \nu_t^{\boldsymbol{\eta}}(y)  -  \mathbb{E}^{\mathbb{Q}}_{Y^{\backslash\boldsymbol{\eta}}_t,\boldsymbol{\lambda}} \big[ \log \Xi_{Y_t,{\boldsymbol{\eta}}} | Y_t^{\boldsymbol{\eta}}=y \big] + l_t^{\boldsymbol{\eta}}(y+1) - l_t^{\boldsymbol{\eta}}(y) \big) - k_t^{\boldsymbol{\eta}}(y) \ ,
\end{align*}
for all $t\in[0,T]$, $\boldsymbol{\eta}\in\mathcal{T}$ and $y\in\mathbb{N}_0$, to be complemented by the slackness conditions $\kappa_t^{\boldsymbol{\eta}}(y)(\nu_t^{\boldsymbol{\eta}}(y)-\bar{\nu})=0$, $\kappa_t^{\boldsymbol{\eta}}(y)\geq 0$ and $\nu_t^{\boldsymbol{\eta}}(y)\leq\bar{\nu}$. By letting $l_t^{\boldsymbol{\eta}}(y)=-\log r_t^{\boldsymbol{\eta}}(y)$ and setting the above expressions to $0$, we obtain
\begin{align*}
\frac{dr_t^{\boldsymbol{\eta}}(y)}{dt} &= r_t^{\boldsymbol{\eta}}(y) \cdot \bigg( \mathbb{E}^{\mathbb{Q}}_{Y^{\backslash\boldsymbol{\eta}}_t,\boldsymbol{\lambda}} \big[ \Xi_{Y_t,{\boldsymbol{\eta}}} | Y_t^{\boldsymbol{\eta}}=y \big] -  \sum_{k=1}^K \delta(t-t_k) \,
\mathbb{E}^{\mathbb{Q}}_{Y^{\backslash\boldsymbol{\eta}}_t} \big[ \log  f_{O|Y_t}(\boldsymbol{o}_k) | Y_t^{\boldsymbol{\eta}}=y \big] \bigg) \\
& - \bigg(1+\frac{k_t^{\boldsymbol{\eta}}(y)}{\phi_t^{\boldsymbol{\eta}}(y)}\bigg) \cdot r_t^{\boldsymbol{\eta}}(y+1) \cdot e^{\mathbb{E}^{\mathbb{Q}}_{Y^{\backslash\boldsymbol{\eta}}_t,\boldsymbol{\lambda}} \big[ \log\Xi_{Y_t,{\boldsymbol{\eta}}} | Y_t^{\boldsymbol{\eta}}=y \big] - k_t^{\boldsymbol{\eta}}(y)/\phi_t^{\boldsymbol{\eta}}(y) }
\end{align*}
and
\begin{align*}
\nu_t^{\boldsymbol{\eta}}(y) =  \frac{r_t^{\boldsymbol{\eta}}(y+1)}{r_t^{\boldsymbol{\eta}}(y)} \cdot e^{\mathbb{E}^{\mathbb{Q}}_{Y^{\backslash\boldsymbol{\eta}}_t,\boldsymbol{\lambda}} \big[ \log\Xi_{Y_t,{\boldsymbol{\eta}}} | Y_t^{\boldsymbol{\eta}}=y \big] - k_t^{\boldsymbol{\eta}}(y)/\phi_t^{\boldsymbol{\eta}}(y) } \ . 
\end{align*}
Observe above that, for fixed values of $\phi$ and $r$, if $$\frac{r_t^{\boldsymbol{\eta}}(y+1)}{r_t^{\boldsymbol{\eta}}(y)} \cdot e^{\mathbb{E}^{\mathbb{Q}}_{Y^{\backslash\boldsymbol{\eta}}_t,\boldsymbol{\lambda}} \big[ \log\Xi_{Y_t,{\boldsymbol{\eta}}} | Y_t^{\boldsymbol{\eta}}=y \big]}<\bar{\nu}$$
then the complementary slackness conditions imply $k_t^{\boldsymbol{\eta}}(y)=0$; otherwise, $\nu_t^{\boldsymbol{\eta}}(y)=\bar{\nu}$ and 
$$k_t^{\boldsymbol{\eta}}(y) = \phi_t^{\boldsymbol{\eta}}(y) \cdot \big[ \mathbb{E}^{\mathbb{Q}}_{Y^{\backslash\boldsymbol{\eta}}_t,\boldsymbol{\lambda}} \big[ \log\Xi_{Y_t,{\boldsymbol{\eta}}} | Y_t^{\boldsymbol{\eta}}=y \big] - \log\frac{\bar{\nu} \cdot r_t^{\boldsymbol{\eta}}(y)}{r_t^{\boldsymbol{\eta}}(y+1)} \big] \geq 0,$$
yielding a valid system of equations, leading to \eqref{ratesEq}-\eqref{jumpR} and concluding the proof.
\end{proof}

\begin{cor}[Distributed network monitoring] \label{Cor1} Assume that network observations are distributed and independent across the stations, so that
\begin{align*}
f_{O|Y_{t_k}}(\boldsymbol{o}_k) = \prod_{i=1}^M f_{O|\{ Y^{\boldsymbol{\eta}}_{t_k}  \, :\, \boldsymbol{\eta}\in\mathcal{T}_i  \}}(\boldsymbol{o}_k^i)
\end{align*} 
for some conditional mass function $f_{O|\{ Y^{\boldsymbol{\eta}}_{t_k}  \, :\, \boldsymbol{\eta}\in\mathcal{T}_i  \}}$, where $\mathcal{T}_i=(\cup_{c\in\mathcal{C}}\mathcal{T}_{i,c}^{\leftarrow}) \, \cup  \, (\cup_{c\in\mathcal{C}}\mathcal{T}_{i,c}^{\rightarrow})$ is the set of job transitions relevant to network activity in station $i>0$, and $\boldsymbol{o}^i_k$ denotes the time $t_k$ observations across classes in the station. Further assume that $\bar{\nu}=\infty$, so that there exists no bound on intensity rates $\nu_t^{\boldsymbol{\eta}}(y)$ under $\mathbb{Q}$. Then, the system of equations in Proposition \ref{Prop1} reduces to
\begin{align*}
\nu_t^{\boldsymbol{\eta}}(y) =  \frac{r_t^{\boldsymbol{\eta}}(y+1)}{r_t^{\boldsymbol{\eta}}(y)} e^{\mathbb{E}^{\mathbb{Q}}_{Y^{\backslash\boldsymbol{\eta}}_t,\boldsymbol{\lambda}} \big[ \log\Xi_{Y_t,{\boldsymbol{\eta}}} | Y_t^{\boldsymbol{\eta}}=y \big] } 
\end{align*}
for all $t\in[0,T]$, $\boldsymbol{\eta}\in\mathcal{T}$ and $y\in\mathbb{N}_0$, with
\begin{align*}
\frac{dr_t^{\boldsymbol{\eta}}(y)}{dt} = r_t^{\boldsymbol{\eta}}(y)  \mathbb{E}^{\mathbb{Q}}_{Y^{\backslash\boldsymbol{\eta}}_t,\boldsymbol{\lambda}} \big[ \Xi_{Y_t,{\boldsymbol{\eta}}} | Y_t^{\boldsymbol{\eta}}=y \big] - r_t^{\boldsymbol{\eta}}(y+1) e^{\mathbb{E}^{\mathbb{Q}}_{Y^{\backslash\boldsymbol{\eta}}_t,\boldsymbol{\lambda}} \big[ \log\Xi_{Y_t,{\boldsymbol{\eta}}} | Y_t^{\boldsymbol{\eta}}=y \big]} 
\end{align*}
whenever $t\neq t_k$, $k=1,\dots,K$, and 
\begin{align*}
\lim_{t\rightarrow t^-_k} r_t^{\boldsymbol{\eta}}(y) = r_{t_k}^{\boldsymbol{\eta}}(y) e^{ \mathbb{E}^{\mathbb{Q}}_{Y^{\backslash\boldsymbol{\eta}}_t} \big[ \log  f_{O|\{ Y^{\boldsymbol{\eta}'}_{t}  \, :\, \boldsymbol{\eta}'\in\mathcal{T}_{\eta_1}  \}}(\boldsymbol{o}_k^{\eta_1}) | Y_t^{\boldsymbol{\eta}}=y \big]} e^{\mathbb{E}^{\mathbb{Q}}_{Y^{\backslash\boldsymbol{\eta}}_t} \big[ \log  f_{O|\{ Y^{\boldsymbol{\eta}'}_{t}  \, :\, \boldsymbol{\eta}'\in\mathcal{T}_{\eta_2}  \}}(\boldsymbol{o}_k^{\eta_2}) | Y_t^{\boldsymbol{\eta}}=y \big]} 
\end{align*} 
accounting for observations at origin and departure nodes in $\boldsymbol{\eta}\in\mathcal{T}$.
\end{cor}

Here, we have obtained a system of equations with iterated dependencies given a distribution $\mathbb{Q}_{\boldsymbol{\lambda}}$. The hazard rate in each counting process depends on the network state probability across the indexed times; complementarily, these state probabilities may be updated independently by means of the master equation \eqref{masterEq}. In Corollary \ref{Cor1}, we further notice that by simplifying the network observation model, and easing restrictions on rates under the approximating measure $\mathbb{Q}$ we retrieve an analogue result to that previously presented in \cite{opper2008variational,cohn2010mean}. However, this is reportedly problematic and can cause a computational bottleneck when reconstructing the jump rates $\nu_t^{\boldsymbol{\eta}}(y)$, as these may approach infinity at observation times. Next, we derive the main result on the infinitesimal transition rates.

\begin{prop} \label{Prop2}
Let densities for the infinitesimal rates $\boldsymbol{\lambda}$ be defined w.r.t to a (Lebesgue) product base measure $\mu_{\boldsymbol{\lambda}}$, so that $d\mathbb{Q}_{\boldsymbol{\lambda}}=g_{\boldsymbol{\lambda}} d\mu_{\boldsymbol{\lambda}}$  with $ g_{\boldsymbol{\lambda}} = \prod_{\boldsymbol{\eta}\in\mathcal{T}} g^{\boldsymbol{\eta}}_\lambda$ and marginal densities $g^{\boldsymbol{\eta}}_{\lambda} = d\mathbb{Q}_{\lambda,\boldsymbol{\eta}}/d\mu_{\lambda}$. Also, let $\nu_t^{\boldsymbol{\eta}}(y)$, $\boldsymbol{\eta}\in\mathcal{T}$, be some (independent) intensity functions assigned to $Y$ under the approximating mean-field measure $\mathbb{Q}$. Finally, define $\boldsymbol{\lambda}_{\backslash\boldsymbol{\eta}}=\{\lambda_{\boldsymbol{\eta}'} \, :\, \boldsymbol{\eta}'\in\mathcal{T} \backslash \{\boldsymbol{\eta}\}  \}$ and recall definitions for network station loads $\Upsilon$ in \eqref{Upsil1} and \eqref{Upsil2}. Then, as $\delta \rightarrow 0$ in \eqref{deltaGenerator}, the distribution $\mathbb{Q}_{\boldsymbol{\lambda}}$ that optimizes the lower bound  \eqref{varboundDeveloped} is such that
\begin{align*}
g_{\lambda}^{\boldsymbol{\eta}}(z) \propto e^{\mathbb{E}^{\mathbb{Q}}_{\boldsymbol{\lambda}_{\backslash\boldsymbol{\eta}}} [   \log f_{\boldsymbol{\lambda}}(\boldsymbol{z}) ] -z \cdot \int_0^T  \mathbb{E}^{\mathbb{Q}}_{Y_t} [ \Upsilon(Y_t,\eta_1,\eta_3)\vee 0  ] dt} \cdot z^{\int_0^T  \mathbb{E}^{\mathbb{Q}}_{Y^{\boldsymbol{\eta}}_t} [   \nu^{\boldsymbol{\eta}}_t (Y^{\boldsymbol{\eta}}_t) ] dt}
\end{align*}
up to a normalizing constant, for $z\in\mathbb{R}_+$ and every $\boldsymbol{\eta}\in\mathcal{T}$.
\end{prop} 

\begin{proof}
We again identify a stationary point to the Lagrangian associated with a constrained optimization problem, w.r.t to arbitrary (positive) densities $g^{\boldsymbol{\eta}}_{\lambda}$ with $ \int_{\mathbb{R}_+} g^{\boldsymbol{\eta}}_\lambda d\mu_{\lambda} =  1 $. Since $\mathbb{P}_{\boldsymbol{\lambda}}$ and $\mathbb{Q}_{\boldsymbol{\lambda}}$ share base measures, the object function can be written as
\begin{align*}
\Phi[g]&= C - \sum_{{\boldsymbol{\eta}}\in\mathcal{T}}  \mathbb{E}^{\mathbb{Q}}_{\lambda_{\boldsymbol{\eta}}} \big[ \Psi[\lambda_{\boldsymbol{\eta}},g^{\boldsymbol{\eta}}_\lambda] \big] -  \sum_{{\boldsymbol{\eta}}\in\mathcal{T}}  l^{\boldsymbol{\eta}} \Big[ \int_{\mathbb{R}_+} g^{\boldsymbol{\eta}}_\lambda d\mu_{\lambda} -  1\Big]
\end{align*}
where $\{l^{\boldsymbol{\eta}}\}_{\boldsymbol{\eta}\in\mathcal{T}}$ are non-functional Lagrange multipliers, and
\begin{align*}
\Psi[\lambda_{\boldsymbol{\eta}},g^{\boldsymbol{\eta}}_\lambda] &=  \log g_{\lambda}^{\boldsymbol{\eta}}  -   \frac{1}{|\mathcal{T}|}  \mathbb{E}^{\mathbb{Q}}_{\boldsymbol{\lambda}_{\backslash\boldsymbol{\eta}}} [   \log f_{\boldsymbol{\lambda}} ]  +   \int_0^T  \mathbb{E}^{\mathbb{Q}}_{Y_t} \bigg[  \nu^{\boldsymbol{\eta}}_t (Y^{\boldsymbol{\eta}}_t) \log \frac{ \nu^{\boldsymbol{\eta}}_t(Y^{\boldsymbol{\eta}}_t)}{\Xi_{Y_t,{\boldsymbol{\eta}}}}  + \Xi_{Y_t,\boldsymbol{\eta}} - \nu^{\boldsymbol{\eta}}_t(Y_t) \bigg] dt . 
\end{align*}
The term $C$ includes the remainder bits in the lower bound in \eqref{varboundDeveloped} that are independent of the rates $\boldsymbol{\lambda}$. It follows that
\begin{align*}
\fdv{\Phi}{g^{\boldsymbol{\eta}}_\lambda} &=  \mathbb{E}^{\mathbb{Q}}_{\boldsymbol{\lambda}_{\backslash\boldsymbol{\eta}}} [   \log f_{\boldsymbol{\lambda}} ]  -\log(g_{\lambda}^{\boldsymbol{\eta}}) - 1 - l^{\boldsymbol{\eta}}   -   \int_0^T  \mathbb{E}^{\mathbb{Q}}_{Y_t} \bigg[  \nu^{\boldsymbol{\eta}}_t (Y^{\boldsymbol{\eta}}_t) \log \frac{ \nu^{\boldsymbol{\eta}}_t(Y^{\boldsymbol{\eta}}_t)}{\Xi_{Y_t,{\boldsymbol{\eta}}}}  + \Xi_{Y_t,\boldsymbol{\eta}} - \nu^{\boldsymbol{\eta}}_t(Y_t) \bigg] dt 
\end{align*}
in its support set $\mathbb{R}_+$, for all $\boldsymbol{\eta}\in\mathcal{T}$. By equating the above to $0$, considering constraints and analysing the relevant terms up to proportionality, we note that
\begin{align*}
g_{\lambda}^{\boldsymbol{\eta}} \propto \exp( \mathbb{E}^{\mathbb{Q}}_{\boldsymbol{\lambda}_{\backslash\boldsymbol{\eta}}} [   \log f_{\boldsymbol{\lambda}} ]  +   \int_0^T  \mathbb{E}^{\mathbb{Q}}_{Y_t} \bigg[   \nu^{\boldsymbol{\eta}}_t (Y^{\boldsymbol{\eta}}_t) \log \Xi_{Y_t,{\boldsymbol{\eta}}}  - \Xi_{Y_t,\boldsymbol{\eta}} \bigg] dt ) ,
\end{align*}
so that
\begin{align*}
g_{\lambda}^{\boldsymbol{\eta}}(z) \propto e^{\mathbb{E}^{\mathbb{Q}}_{\boldsymbol{\lambda}_{\backslash\boldsymbol{\eta}}} [   \log f_{\boldsymbol{\lambda}}(\boldsymbol{z}) ] - \int_0^T z \cdot \mathbb{E}^{\mathbb{Q}}_{Y_t} [ \Upsilon(Y_t,\eta_1,\eta_3)\vee 0  ] dt + \int_0^T  \mathbb{E}^{\mathbb{Q}}_{Y^{\boldsymbol{\eta}}_t} [   \nu^{\boldsymbol{\eta}}_t (Y^{\boldsymbol{\eta}}_t) \log (\delta + z\cdot\Upsilon(Y_t,\eta_1,\eta_3)) ] dt}
\end{align*}
and
\begin{align*}
g_{\lambda}^{\boldsymbol{\eta}}(z) \propto e^{\mathbb{E}^{\mathbb{Q}}_{\boldsymbol{\lambda}_{\backslash\boldsymbol{\eta}}} [   \log f_{\boldsymbol{\lambda}}(\boldsymbol{z}) ] -z \cdot \int_0^T  \mathbb{E}^{\mathbb{Q}}_{Y_t} [ \Upsilon(Y_t,\eta_1,\eta_3)\vee 0  ] dt} \cdot z^{\int_0^T  \mathbb{E}^{\mathbb{Q}}_{Y^{\boldsymbol{\eta}}_t} [   \nu^{\boldsymbol{\eta}}_t (Y^{\boldsymbol{\eta}}_t) ] dt}
\end{align*}
as $\delta\rightarrow 0$, for $z\in\mathbb{R}_+$ and every $\boldsymbol{\eta}\in\mathcal{T}$.
\end{proof}

\begin{cor} [Conjugate prior] \label{Cor2} 
Assume that the prior density on $\boldsymbol{\lambda}$ also factors across the individual rates, s.t $d\mathbb{P}_{\boldsymbol{\lambda}}= \prod_{\boldsymbol{\eta}\in\mathcal{T}} f^{\boldsymbol{\eta}}_\lambda d\mu_{\boldsymbol{\lambda}}$, where $f_{\lambda}^{\boldsymbol{\eta}}$ for $\boldsymbol{\eta}\in\mathcal{T}$ denote Gamma density functions with shape $\alpha_{\boldsymbol{\eta}}$ and rate $\beta_{\boldsymbol{\eta}}$. Then, as $\delta \rightarrow 0$ in \eqref{deltaGenerator}, these are conjugate priors and
$$
\lambda_{\boldsymbol{\eta}} \stackrel{\mathbb{Q}}{\sim}\Gamma\Big(\alpha_{\boldsymbol{\eta}} + \int_0^T  \mathbb{E}^{\mathbb{Q}}_{Y^{\boldsymbol{\eta}}_t} [   \nu^{\boldsymbol{\eta}}_t (Y^{\boldsymbol{\eta}}_t) ] dt
, \beta_{\boldsymbol{\eta}} + \int_0^T  \mathbb{E}^{\mathbb{Q}}_{Y_t} [ \Upsilon(Y_t,\eta_1,\eta_3)\vee 0  ] dt \Big),
$$
for all $\boldsymbol{\eta}\in\mathcal{T}$.
\end{cor}

Hence, as the network model with negative loads offers a better approximation of its original counterpart, we may numerically approximate posterior distributions across the infinitesimal rates in $\boldsymbol{\lambda}$, under the mean-field measure $\mathbb{Q}$. In the special case with independent Gamma prior densities, this is an easily interpretable posterior where the shape and rate parameters depend, respectively, on the integrated expected jump intensities and the integrated expected station loads.

\section{Applications} \label{exmpls}

The results in this paper suggest an iterative approximation procedure to \eqref{likRatio} by means of coordinate ascent. Here, we iteratively update the values of the various rates, functions and Lagrange multipliers while evaluating, and assessing convergence, in the lower bound \eqref{varboundDeveloped} to the log-likelihood. This is a standard approach in variational inference when looking for a (local) maxima \citep{blei2017variational}, and the problem is known to be convex.

Maximising the bound by calibrating the measure $\mathbb{Q}$ will yield an approximation to the regular conditional probability of events in $\mathcal{F}$ under $\tilde{\mathbb{P}}$, conditioned on the observations. Projected over the rates $\boldsymbol{\lambda}$,  it yields an approximation to the posterior rate density and the likelihood ratio in \eqref{likRatio}. This projected densities are valid in order to approximate the conditional distributions of the service rates $\boldsymbol{\mu}$ and routing probabilities $\mathcal{P}$ in the original queueing network system, given observationa. The final iterative procedure is described here.

\begin{itemize}
\item First, input network observations in \eqref{obsLik} and assign a (conjugate) Gamma (image) density $d\mathbb{P}_{\boldsymbol{\lambda}}$ across job transition intensities $\boldsymbol{\lambda}$, with shape parameters $\alpha_{\boldsymbol{\eta}}$ and a (shared) rate parameter $\beta_{\boldsymbol{\eta}}=\beta$, $\boldsymbol{\eta}\in\mathcal{T}$. 
\item Define a discretization grid of the time interval $[0,T]$, and operate through interpolation within points in the grid. 
\item Set an arbitrary density $d\mathbb{Q}_{\boldsymbol{\lambda}}$. Fix $\kappa_t^{\boldsymbol{\eta}}(y)=0,r_t^{\boldsymbol{\eta}}(y)=1$ and input (valid) arbitrary starting values $\mathbb{Q}(Y^{\boldsymbol{\eta}}_t=y)$, for all $t\in[0,T]$, $\boldsymbol{\eta}\in\mathcal{T}$ and $y\in\mathbb{N}_0$.
\item Iterate until convergence:
\begin{itemize}
\item In each direction $\boldsymbol{\eta}\in\mathcal{T}$, numerically compute $r_t^{\boldsymbol{\eta}}(y)$ for every $t\in[0,T]$ and $y\in\mathbb{N}_0$, by means of \eqref{diffR}-\eqref{jumpR}. Then, update intensity and slack functions $\nu_t^{\boldsymbol{\eta}}(y),\kappa_t^{\boldsymbol{\eta}}(y)$ with \eqref{ratesEq}, so that $k_t^{\boldsymbol{\eta}}(y)=0$ if
$$\frac{r_t^{\boldsymbol{\eta}}(y+1)}{r_t^{\boldsymbol{\eta}}(y)} \cdot e^{\mathbb{E}^{\mathbb{Q}}_{Y^{\backslash\boldsymbol{\eta}}_t,\boldsymbol{\lambda}} \big[ \log\Xi_{Y_t,{\boldsymbol{\eta}}} | Y_t^{\boldsymbol{\eta}}=y \big]}<\bar{\nu} ,$$
and $$\nu_t^{\boldsymbol{\eta}}(y) = \bar{\nu}, \quad k_t^{\boldsymbol{\eta}}(y) = \mathbb{Q}(Y_t^{\boldsymbol{\eta}}=y) \cdot \big[ \mathbb{E}^{\mathbb{Q}}_{Y^{\backslash\boldsymbol{\eta}}_t,\boldsymbol{\lambda}} \big[ \log\Xi_{Y_t,{\boldsymbol{\eta}}} | Y_t^{\boldsymbol{\eta}}=y \big] - \log\frac{\bar{\nu} \cdot r_t^{\boldsymbol{\eta}}(y)}{r_t^{\boldsymbol{\eta}}(y+1)} \big]$$ otherwise. Renew transient state probabilities in $Y$ by means of the master equation \eqref{masterEq}, for all $t\in [0,T]$.
\item Derive expected jump intensities $\mathbb{E}^{\mathbb{Q}}_{Y^{\boldsymbol{\eta}}_t} [   \nu^{\boldsymbol{\eta}}_t (Y^{\boldsymbol{\eta}}_t) ]$ and 
station loads $\mathbb{E}^{\mathbb{Q}}_{Y_t} [ \Upsilon(Y_t,\eta_1,\eta_3)\vee 0  ]$, for all directions $\boldsymbol{\eta}\in\mathcal{T}$ and times $t\in[0,T]$. Update $\mathbb{Q}$-densities for rates s.t.
$$
\lambda_{\boldsymbol{\eta}} \stackrel{\mathbb{Q}}{\sim}\Gamma\Big(\alpha_{\boldsymbol{\eta}} + \int_0^T  \mathbb{E}^{\mathbb{Q}}_{Y^{\boldsymbol{\eta}}_t} [   \nu^{\boldsymbol{\eta}}_t (Y^{\boldsymbol{\eta}}_t) ] dt
, \beta + \int_0^T  \mathbb{E}^{\mathbb{Q}}_{Y_t} [ \Upsilon(Y_t,\eta_1,\eta_3)\vee 0  ] dt \Big),
$$
for all $\boldsymbol{\eta}\in\mathcal{T}$. The likelihood ratio in \eqref{likRatio} can be computed at this stage.
\item Evaluate the lower bound \eqref{varboundDeveloped}, given the current densities and infinitesimal rates under the approximating measure $\mathbb{Q}$. Assess variation in the bound across iterations and establish convergence.
\end{itemize}
\item Finally, infer the structure of the various service rates and routing probabilities in the queueing network system. 
\begin{itemize}
\item Note that $\mathbb{E}^{\mathbb{Q}}_{Y_t} [ \Upsilon(Y_t,\eta_1,\eta_3)\vee 0]$ remains the same across directions $\boldsymbol{\eta}\in\mathcal{T}$ with shared origin station. Since $\mu_i^c = \sum_{\boldsymbol{\eta}\in\mathcal{T}^{\rightarrow}_{i,c}} \lambda_{\boldsymbol{\eta}}$ it holds
$$\textstyle  \mu_i^c \stackrel{\mathbb{Q}}{\sim}\Gamma\Big(|\mathcal{T}^{\rightarrow}_{i,c}|\cdot\alpha_{\boldsymbol{\eta}} + \sum_{\boldsymbol{\eta}\in\mathcal{T}^{\rightarrow}_{i,c}} \int_0^T  \mathbb{E}^{\mathbb{Q}}_{Y^{\boldsymbol{\eta}}_t} [   \nu^{\boldsymbol{\eta}}_t (Y^{\boldsymbol{\eta}}_t) ] dt
, \beta + \int_0^T  \mathbb{E}^{\mathbb{Q}}_{Y_t} [ \Upsilon(Y_t,\eta_1,\eta_3)\vee 0  ] dt \Big),$$
for all $ 0\leq i\leq M, c\in\mathcal{C}$.
\item Retrieve distributions for routing probabilities by noting that $p_{i,j}^c=\lambda_{i,j,c}/\mu_i^c$ for all $ 0\leq i,j\leq M, c\in\mathcal{C}$. This suggests a Dirichlet distribution.
\end{itemize}
\end{itemize}

The following examples treat open and closed network models; source code can be found at \href{https://github.com/IkerPerez/variationalQueues}{github.com/IkerPerez/variationalQueues.}

\subsection{Single class closed network}

We begin with a small closed network example as shown in Figure \ref{exmpl1}; this includes one FCFS service station, with $K_1=1$ processing unit, along with a delay node, together processing a population of $N$ jobs cyclically in a closed loop. All jobs belong to the same class and have equal service rates, we denote by $\mu_1$ the job processing rate within the service station. On completion, a job proceeds to the delay node where it awaits for an exponentially distributed time before being routed back to the queue. We use $\mu_0$ to denote the delay rate; and note that the arrival rate to the queue is directly proportional to the number of jobs at the delay.
\begin{figure}[h!]
\centering
\begin{tikzpicture}
% Network 1
\draw (0,0) -- ++(1.4cm,0) -- ++(0,-0.8cm) -- ++(-1.4cm,0);
\foreach \i in {1,...,4}
  \draw (1.4cm-\i*8pt,0) -- +(0,-0.8cm);
\filldraw[fill=black!05!white] (1.4cm+0.41cm,-0.4cm) circle [radius=0.4cm]; % Circle

% Delays...
\filldraw[fill=black!05!white] (-1.9cm,0.2cm) circle [radius=0.3cm]; % Circle
\filldraw[fill=black!05!white] (-1.9cm,-1cm) circle [radius=0.3cm]; % Circle

% the arrows and labels
\draw[->,>=latex,line width=0.20mm] (2.3cm,-0.4cm) -- ++(1cm,0cm) -- ++(0cm,-1.6cm) -- ++(-7cm,0cm) -- ++(0cm,1.6cm) --
	++(1cm,0cm);
\draw[<-,>=latex,line width=0.20mm] (-0.1cm,-0.4cm) -- ++(-1cm,0cm);
\draw[-,line width=0.20mm] (-2.7cm,-0.4cm) -- (-2.2cm,0.2cm);
\draw[-,line width=0.20mm] (-2.7cm,-0.4cm) -- (-2.2cm,-1cm);
\draw[-,line width=0.20mm] (-1.6cm,0.2cm) -- (-1.1cm,-0.4cm);
\draw[-,line width=0.20mm] (-1.6cm,-1cm) -- (-1.1cm,-0.4cm);
\node at (-1.9,-0.3cm) {$\vdots$};

% Text superimposed
\node at (1.83,-0.4cm) {$\mu_1$};
\node at (-1.9,0.2cm) {\small $\mu_0$};
\node at (-1.9,-1cm) {\small $\mu_0$};
\node at (-1.9,-1.6cm) {\small Delay};
\node at (1,-1.15cm) {\small Service station};

% Priors...
%\draw[-,line width=0.20mm, dashed] (2.27cm,-0.2cm) -- ++ (1.58,0);
%\draw[-,line width=0.20mm, dashed] (-2.25cm,0.3cm) -- ++ (-1.8,0);
%\node at (-4.55,0.3cm) {\small $\alpha_\lambda,\beta_\lambda$};
%\node at (4.3,-0.2cm) {\small $\mathrm{d}\mathbb{P}_{\mu_1}$};
\end{tikzpicture}
\vspace{10pt}
\caption{Closed queueing network with a single FCFS service station and a delay.} \label{exmpl1}
\end{figure}
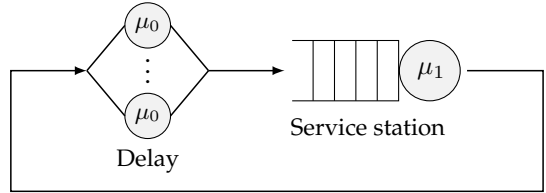
 
Both nodes are independent and $\mu_0$ is fixed in order to ensure model identifiability within the service station. In this instance, the network topology is deterministic and trivial, and the evolution of $X=(X_t)_{t\geq 0}$ monitors the total number of jobs within the service station, with $X_0=0$. The generator $Q$ of $X$ is finite and s.t.
$$Q=\kbordermatrix{
      & 0 & 1 & 2 & \dots & N-2 & N-1 & N \\
    0 & -N\mu_0 & N\mu_0 & 0 & \dots & 0 & 0 & 0\\
    1 & \mu_1 & -N\mu_0+\mu_0-\mu_1 & (N-1)\mu_0 & \dots & 0 & 0 & 0 \\
    \vdots & \vdots & \vdots & \vdots & \ddots & \vdots & \vdots & \vdots\\
    N-1 & 0 & 0 & 0 & \dots & \mu_1 & -(\mu_0+\mu_1) & \mu_0 \\
    N & 0 & 0 & 0 & \dots & 0 & \mu_1 & -\mu_1
  },$$
where row and column labels denote the number of jobs in the queueing node. Since $\mu_0$ is fixed, our interest lies in $\lambda\equiv \lambda_{1,0} = \mu_1 \cdot  p_{1,0} = \mu_1$. We assign to this rate a distribution $\mathbb{P}_\lambda\equiv \lambda_\star\mathbb{P}$ with (Gamma) density $f_\lambda$ such that its hyperparameters fix some reasonably uninformative \textit{prior} knowledge on the system.  We monitor the delay node and FCFS service station at fixed and equally spaced times $t_1<\dots<t_K$ in an interval $[0,T]$. Here, variables $O_k$ are supported on $\mathcal{O}=\{0,\dots,N\}^2$ and the observation model factors across the network components s.t.
$f_{O|x}(\boldsymbol{o}) = \tilde{f}_{O|N-x}(o_0) \cdot \tilde{f}_{O|x}(o_1)$ with $\tilde{f}_{O|x}(o)=\frac{\epsilon}{N} + \mathbb{I}(o=x) \cdot (1-\frac{N+1}{N}\epsilon)$ and
\begin{align}
\mathbb{P}(O_k^{-1}(\boldsymbol{o})|X)=
\begin{cases}
               (1-\epsilon)^2               & o_0= N-X_{t_k},\, o_1 = X_{t_k},\\
               (\epsilon/N)^2  & o_0 \neq N-X_{t_k},\, o_1 \neq X_{t_k}, \\
   			   (1-\epsilon)\cdot\epsilon/N  & \text{otherwise}.
           \end{cases} \label{obsLike}           
\end{align}
for $\epsilon>0$ and all $k=1,\dots,K$. This accounts for some $\%100\cdot\epsilon$ faulty measurements, also, we note that a system with discrete observations is approximated as $\epsilon\rightarrow 0$. Now, assume there exist some sample observations $\boldsymbol{o}_1,\dots,\boldsymbol{o}_K$ from a model realization in this closed network. These can be easily produced from \eqref{obsLike} given a trajectory $(X_t)_{t\in[0,T]}$. In order to produce the trajectory from \eqref{densPath}, given the service rates, we may employ Gillespie's algorithm \citep{gillespie1977exact} or faster \textit{uniformization} based alternatives \citep{rao13a}. 

\noindent \textbf{Remark.} The transformation $d\mathbb{P}_{\lambda|\boldsymbol{o}_1,\dots,\boldsymbol{o}_K}/d\mathbb{P}_{\lambda}$ is such that, conditioned on $\boldsymbol{o}_1,\dots,\boldsymbol{o}_K$, the distribution over $\lambda$ admits a density carried by a Lebesgue measure $\mu_\lambda$, so that $d\mathbb{P}_{\lambda|\boldsymbol{o}_1,\dots,\boldsymbol{o}_K}=f_{\lambda|\boldsymbol{o}_1,\dots,\boldsymbol{o}_K}\,d\mu_\lambda$. In this simple example, numerical MCMC procedures \citep{perez2017auxiliary} or basic \textit{generator-matrix exponentiations} combined with a forward-backward algorithm \citep{zhao2016bayesian} can offer such density approximations; however, this is reportedly \textit{very inefficient} when $N$ is large. Moreover, in involved networks/processes for complex applications (see next example), such alternatives are simply \textit{unusable} (i.e. they do not scale).

In the following, we analyse simulated data ($N=50,\,\mu_0=0.1,\,\epsilon=0.2,\,T=100,\,K=50$) by assigning a conjugate Gamma density to $\lambda\equiv\mu_1$, so that $\lambda\sim\Gamma(\alpha,\beta)$ with $\alpha=5$ and $\beta=2$ under the reference measure $\mathbb{P}$. Recall that $\mu_0$ is fixed to ensure model identifiability, and $X_t\in\{0,\dots,N\}$ denotes the number of jobs in the service station at any time $t\geq0$. For later reference, the stationary distribution of the system is given by
\begin{align*} 
\pi_{\mathbb{P}}(x|\lambda) = \lim_{t\rightarrow\infty}  \mathbb{P}(X_t=x|\lambda) = (\frac{\mu_0}{\lambda})^x\frac{1}{(N-x)!}\bigg/\sum_{x=0}^N(\frac{\mu_0}{\lambda})^x\frac{1}{(N-x)!}, 
\end{align*}
so that, assuming the observations are sufficiently spaced, and that the system has reached stationarity, it holds
\begin{align} 
\mathbb{P}\Big(\bigcap_{k=1}^K O_k^{-1}(\boldsymbol{o}_k)\big|\lambda\Big) \approx \prod_{k=1}^K \sum_{x=0}^N f_{O|x}(\boldsymbol{o}_k)  \pi_{\mathbb{P}}(x|\lambda) = \frac{\prod_{k=1}^K \sum_{x=0}^N f_{O|x}(\boldsymbol{o}_k)  (\frac{\mu_0}{\lambda})^x\frac{1}{(N-x)!}}{\Big(\sum_{x=0}^N(\frac{\mu_0}{\lambda})^x\frac{1}{(N-x)!}\Big)^K}, \label{LikelihoodStationary}
\end{align}
where $f_{O|x}(\boldsymbol{o}_k)$ is as defined in \eqref{obsLike}. Note that here $|\mathcal{T}|=2$, and the process $Y=(Y^{0,1}_t,Y^{1,0}_t)_{t\geq 0}$ monitors transitions between the delay and service station, in both the directions $0\rightarrow 1$ and $1\rightarrow 0$. The lower bound to the log-likelihood in \eqref{varboundDeveloped} reduces to
\begin{align}
\log\tilde{\mathbb{P}}(\boldsymbol{O}) \geq \sum_{k=1}^K\mathbb{E}^{\mathbb{Q}}_{Y_{t_k}} \big[ \log  \mathbb{P}(O^{-1}_k(\boldsymbol{o}_k) |Y^{0,1}_{t_k} - Y^{1,0}_{t_k})  \big]
 -\mathbb{E}^{\mathbb{Q}}_{\lambda} \bigg[  \log\frac{ g_{\lambda}^{1,0} }{ d\mathbb{P}_{\lambda}   } \bigg] + \int_0^T  \mathbb{E}^{\mathbb{Q}}_{Y_t,\lambda} \big[ \Psi[Y_t,\nu_t,\lambda] \big] dt  \label{BoundExample1}
\end{align}
with
\begin{align*}
\Psi[Y_t,\nu_t,\lambda] &= 
 \nu^{1,0}_t (Y^{1,0}_t) + \nu^{0,1}_t (Y^{0,1}_t)  - 2\delta - \lambda\cdot\mathbb{I}(Y^{0,1}_{t_k} - Y^{1,0}_{t_k} > 0) - \mu_0\cdot (N + Y^{1,0}_{t_k} - Y^{0,1}_{t_k})   \\
 & -  \nu^{1,0}_t (Y^{1,0}_t) \log \frac{  \nu^{1,0}_t (Y^{1,0}_t)}{\delta + \lambda\cdot\mathbb{I}(Y^{0,1}_{t_k} - Y^{1,0}_{t_k} > 0)} - \nu^{0,1}_t (Y^{0,1}_t) \log \frac{ \nu^{0,1}_t (Y^{0,1}_t) }{\delta + \mu_0\cdot (N + Y^{1,0}_{t_k} - Y^{0,1}_{t_k}) }   
\end{align*}
s.t. it contains only two hazard functions in the approximating measure $\mathbb{Q}$, namely $\nu^{0,1}$ and $\nu^{1,0}$. In \eqref{BoundExample1}, we again notice that the lower bound is dominated by $3$ distinguishable components, i.e. (i) the expected log-observations, (ii) the \textit{Kullback-Leibler} divergence across the service rate density, and (iii) a weighted $\mathbb{P}$-to-$\mathbb{Q}$ divergence in the expected path likelihood, further integrated along the entire network trajectory. The differential equations for functionals in \eqref{jumpR} reduce to
\begin{align*}
\frac{dr_t^{0,1}(y)}{dt} &= r_t^{0,1}(y) \big(\delta + \mu_0 \cdot \mathbb{E}^{\mathbb{Q}}_{Y^{1,0}_t} \big[ (Y^{1,0}_t - y)\vee 0)  \big] \big) \\ 
&- \frac{1+ k_t^{0,1}(y)/\mathbb{Q}(Y_t^{0,1}=y)}{e^{ k_t^{0,1}(y)/\mathbb{Q}(Y_t^{0,1}=y)}}  r_t^{0,1}(y+1) e^{\mathbb{E}^{\mathbb{Q}}_{Y^{1,0}_t} \big[ \log(\delta + \mu_0\cdot[(Y^{1,0}_t-y)\vee 0]) \big]},
\end{align*}
and
\begin{align*}
\frac{dr_t^{1,0}(y)}{dt} &= r_t^{1,0}(y)  \big(\delta + \mathbb{E}^{\mathbb{Q}}_{\lambda}[\lambda]\cdot \mathbb{Q}(Y^{0,1}_t>y) \big) \\
&- \frac{1+k_t^{1,0}(y)/\mathbb{Q}(Y_t^{1,0}=y)}{e^{ k_t^{1,0}(y)/\mathbb{Q}(Y_t^{1,0}=y)}}  r_t^{1,0}(y+1) e^{\mathbb{E}^{\mathbb{Q}}_{Y^{0,1}_t,\boldsymbol{\lambda}} \big[ \log(\delta + \lambda \cdot \mathbb{I}(Y^{0,1}_t>y)) \big]} .
\end{align*}

\begin{figure}[h!]
  \centering
    \includegraphics[width=0.98\textwidth]{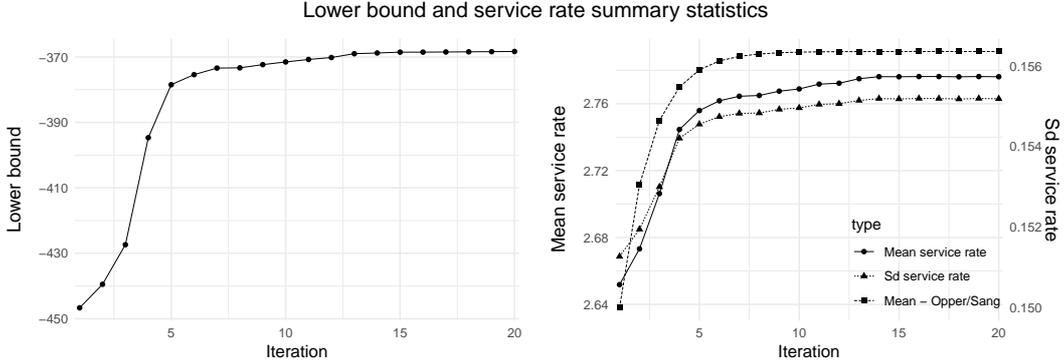}   
   \caption{Left, evolution of lower bound to log-likelihood during the inferential procedure. Right, evolution of mean and standard deviation values for $\lambda$ under $\mathbb{Q}$, along with point estimates from a traiditonal variational procedure.} \label{Example1Bound}
\end{figure}

In Figure \ref{Example1Bound} (left) we observe the evolution of the lower bound \eqref{BoundExample1} during the iterative inferential procedure, for a sufficiently small and negligible value of $\delta$. There, we notice that the procedure has converged to a (local) optima within approximately $13$ iterations. On the right hand side of the Figure, we further observe summary statistics (mean and standard deviation) for $\lambda$ under the approximating measure $\mathbb{Q}$; along with iterative estimations (point estimates) obtained from employing variational procedures in \cite{opper2008variational}. Next, in Figure \ref{outputQueue1} (left) we find the $\mathbb{P}$-prior density for $\lambda$ with the $\mathbb{Q}$-posterior superimposed (in gray); along with them, the red/blue densities represent 
\begin{itemize}
\item posterior density through \textit{Metropolis-Hastings} Markov chain Monte Carlo, by means of strong stationarity assumptions leading to the likelihood function shown in \eqref{LikelihoodStationary},
\item and approximate density extracted by adapting variational procedures in \cite{opper2008variational}, to allow for prior knowledge and conjugacy properties.
\end{itemize}
\begin{figure}[h!]
  \centering
    \includegraphics[width=0.98\textwidth]{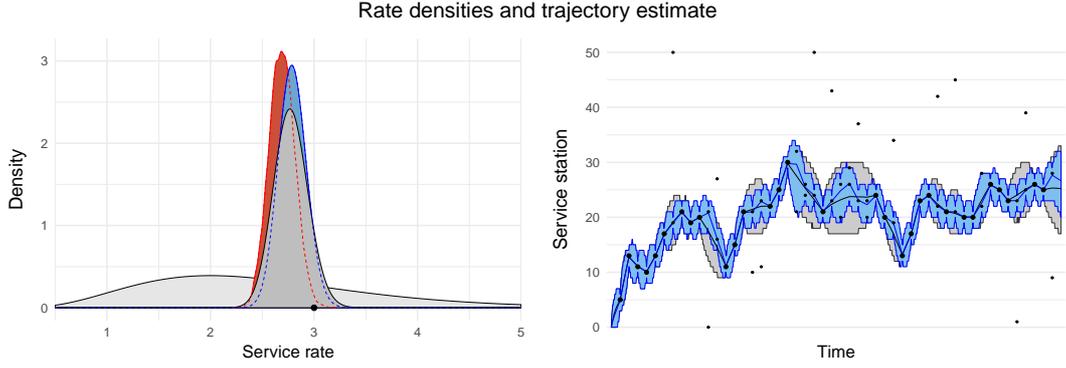}   
   \caption{Left, prior and posterior densities for $\lambda$; along with MCMC and traditional variational density estimates. The black dot on the horizontal axis represents the original value in the network simulation. Right, network observations along with mean-average network trajectory and $95\%$ credible interval for job counts in the service station; in gray, our proposed method, in blue, existing variational alernative method.} \label{outputQueue1}
\end{figure}
On the right hand side we observe the network observations on both the service station and delay node. Delay node observations are displayed by subtracting their value from the job population $N$ (thus representing a second measurement on the service station). Whenever both observations match, these are displayed with a large-sized dot. Along with it, we find:
\begin{itemize}
\item In gray, a mean-average network trajectory and $95\%$ credible interval for job counts on the service station $X_t=(Y^{0,1}_t-Y^{1,0}_t)_{t\geq 0}$, under the approximating measure $\mathbb{Q}$ and with methods introduced in this paper.
\item In blue, a similar confidence interval and mean-average path obtained using benchmark methods in \cite{opper2008variational}.
\end{itemize}
Noticeably, the average mean-field trajectory for our proposed variational technique flows through the most informative observations (thick dots), and the credible interval widens-up to account for some faulty measurements within either network node. On the other hand, traditional variational approaches quickly converge to a local optima, and restrict mean-average dynamics further compressing confidence intervals in regions with noisy data. In the next example, we notice how this poses a problem for traditional methods; that is, within complex and \textit{synchronized} stochastic processes we will fail to obtain sensible estimates for network parameters. Moreover, inference by MCMC/forward-backwards methods in our next example is virtually intractable \citep[cf.][]{perez2017auxiliary}.

\subsection{Multi-class parallel tandems with bottleneck and service priorities} \label{mainExample}

We analyse an open multi-class queueing network as pictured in Figure \ref{exmpl2}. In this network, there exists two classes ($c=1,2$) of jobs that simultaneously transit the system. The first class consists of high priority jobs with low arrival and service intensity rates. The second class includes low priority jobs with high arrival and service rates. Once a job enters the system, a probabilistic routing junction (pictured as a square within the Figure) sends this job through either a \textit{PS} or \textit{priority-FCFS} tandem; later, it will be serviced within an \textit{infinite} node before leaving the network. In the top processor-sharing tandem, each station has 5 processing units; these will fraction their working capacity as seen in \eqref{Upsil1}, in order simultaneously service all jobs regardless of their class and priority level, however, service rates will differ depending on the job class. On the contrary, the bottom tandem includes two \textit{FCFS} stations with a single processing unit and priority scheduling. Within these nodes, low priority jobs are \textit{only} serviced if each station is \textit{fully empty} of any high priority jobs; consequently, station loads in \eqref{Upsil2} are rewritten s.t.
\begin{align*}
\textstyle\textstyle\Upsilon(y,i,1) = 1 \wedge x_{i,1}   \quad \text{and} \quad \textstyle\textstyle\Upsilon(y,i,2) = (1 \wedge x_{i,2}) \cdot \mathbb{I}(x_{i,1} <1),
\end{align*}
at stations $i\in\{2,4\}$ and for any $y\in\mathcal{S}^Y$, where we recall $$x_{i,c} = \sum_{{\boldsymbol{\eta}}\in\mathcal{T}^{\leftarrow}_{i,c}}  y_{{\boldsymbol{\eta}}} -  \sum_{{\boldsymbol{\eta}}\in\mathcal{T}^{\rightarrow}_{i,c}} y_{{\boldsymbol{\eta}}} $$ and thus $x_{2,c}=y_{0,2,c}-y_{2,4,c}, x_{4,c}=y_{2,4,c}-y_{4,5,c}$, for $c=1,2$. Due to the presence of service priorities, the ordering of jobs within the queue is irrelevant (this is also the case with \textit{random order} disciplines); hence, our inferential framework allows for different service rates assigned to jobs in each class. Finally, the last service node includes an infinite amount of processing units, and processing rates also differ depending on the job class.
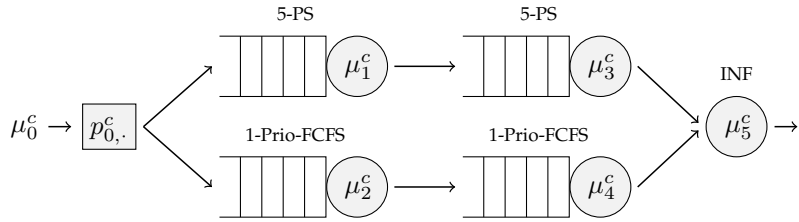
\begin{figure}[h!]
\centering
\begin{tikzpicture}
% Network 1
\draw (0,0.1) -- ++(1.4cm,0) -- ++(0,-0.8cm) -- ++(-1.4cm,0);
\foreach \i in {1,...,4}
  \draw (1.4cm-\i*8pt,0.1) -- +(0,-0.8cm);
\filldraw[fill=black!05!white] (1.4cm+0.41cm,-0.3cm) circle [radius=0.4cm]; % Circle
\node at (1,0.4cm) {\scriptsize 5-PS};

\draw (-3.2,0.1) -- ++(1.4cm,0) -- ++(0,-0.8cm) -- ++(-1.4cm,0);
\foreach \i in {1,...,4}
  \draw (-1.8cm-\i*8pt,0.1) -- +(0,-0.8cm);
\filldraw[fill=black!05!white] (-1.8cm+0.41cm,-0.3cm) circle [radius=0.4cm]; % Circle
\node at (-2.2,0.4cm) {\scriptsize 5-PS};

\draw (0,-1.5cm) -- ++(1.4cm,0) -- ++(0,-0.8cm) -- ++(-1.4cm,0);
\foreach \i in {1,...,4}
  \draw (1.4cm-\i*8pt,-1.5cm) -- +(0,-0.8cm);
\filldraw[fill=black!05!white] (1.4cm+0.41cm,-1.4cm-0.5cm) circle [radius=0.4cm]; % Circle
\node at (1,-1.2cm) {\scriptsize 1-Prio-FCFS};

\draw (-3.2,-1.5cm) -- ++(1.4cm,0) -- ++(0,-0.8cm) -- ++(-1.4cm,0);
\foreach \i in {1,...,4}
  \draw (-1.8cm-\i*8pt,-1.5cm) -- +(0,-0.8cm);
\filldraw[fill=black!05!white] (-1.8cm+0.41cm,-1.4cm-0.5cm) circle [radius=0.4cm]; % Circle
\node at (-2.2,-1.2cm) {\scriptsize 1-Prio-FCFS};

\filldraw[fill=black!05!white] (3.6,-0.7cm-0.4cm) circle [radius=0.4cm]; % Circle
\node at (3.6,-0.4cm) {\scriptsize INF};

\filldraw[fill=black!05!white] (-5cm,-0.8cm) rectangle (-4.3cm,-1.4cm);

% the arrows and labels
\draw[->,line width=0.20mm] (-4.2cm,-1.1cm) -- (-3.3,-0.3cm);
\draw[->,line width=0.20mm] (-4.2cm,-1.1cm) -- (-3.3,-1.9cm);
\draw[->,line width=0.20mm] (-0.9cm,-0.3cm) -- (-0.1,-0.3cm);
\draw[->,line width=0.20mm] (-0.9cm,-1.9cm) -- (-0.1,-1.9cm);
\draw[<-,line width=0.20mm] (-5.15cm,-1.1cm) -- +(-0.3cm,0cm) node[left] {$\mu^c_0$};
\draw[->,line width=0.20mm] (2.3cm,-0.3cm) -- (3.1cm,-1.05cm);
\draw[->,line width=0.20mm] (2.3cm,-1.9cm) -- (3.1cm,-1.15cm);
\draw[->,line width=0.20mm] (4.1cm,-1.1cm) -- +(0.3cm,0cm);

\node at (-1.37,-0.3cm) {$\mu_1^c$};
\node at (-1.37,-1.9cm) {$\mu_2^c$};
\node at (1.83,-0.3cm) {$\mu_3^c$};
\node at (1.83,-1.9cm) {$\mu_4^c$};
\node at (3.63,-1.1cm) {$\mu_5^c$};
\node[align=center] at (-4.65cm,-1.15cm) {$p^c_{0,\cdot}$};
\end{tikzpicture}
\vspace{10pt}
\caption{Open queueing network with one routing juntion (pictured as a square) and 5 service stations with varied disciplines and processing rates.} \label{exmpl2}
\end{figure}

We analyse synthetic data created during a time interval $[0,T]$ ($T=100$), with arrival intensities $\mu_0^1=0.5,\mu_0^2=3$, routing probabilities $p^c_{0,i}=0.5$, $i,c\in\{1,2\}$ and service rates as shown in Table \ref{summaryResults}. We collect a reduced set of \textit{noiseless} and equally spaced observations with $K=50$; these are essentially snapshots of the full system state across its service stations and job classes, s.t. $\mathcal{O}=\mathbb{N}^{10}$ and the observation density in \eqref{obsLik} is defined with
$$f_{O|x}(\boldsymbol{o}) = \prod_{i=1}^5 \prod_{c=1}^2 \mathbb{I}(x_{i,c} = o_{i,c})$$
for $x\in\mathcal{S}$, where $o_{i,c}$ is an indexed observation in the element $\boldsymbol{o}$ denoting the class-$c$ queue length at station $i>0$ \footnote{Source code for the data simulation process may be found at \href{https://github.com/IkerPerez/variationalQueues/tree/master/Example 2}{github.com/IkerPerez/variationalQueues.}}. Within the inferential procedure, this observation likelihood must be approximated with some regularized variant similar to \eqref{obsLike}, while taking $\epsilon\rightarrow 0$. Next, we assign conjugate Gamma priors to the various service intensities; in order to ensure identifiability in the problem, arrival rates and routing probabilities are fixed and we focus this inferential task on the various service stations. Hence $\boldsymbol{\lambda}\equiv \{\mu^c_i\, :\, c=1,2 \; \text{and} \; i=1,\cdots,5\}$, and we set $\lambda_{\boldsymbol{\eta}}\sim\Gamma(1,0.3)$ under the reference measure $\mathbb{P}$, for all $\boldsymbol{\eta}\in\mathcal{T}$.  

In the following, we omit the cumbersome mathematical details related to this complex model formulation, and we focus on discussing prior choices, calibration of the algorithm, results and method comparisons following the inferential procedure. 

\noindent\textbf{Remarks on using MCMC data-augmentation for inference.} Transient inference in a stochastic system with priorities is specially challenging, due to the strong dependencies this generates on the queue lengths across the nodes and classes. Specifically,
\begin{itemize}
\item data-augmentation methods relying on MCMC techniques do not scale \citep[cf.][]{sutton2011,perez2017auxiliary}, as dependences yield very autocorrelated output chains,
\item there exist no analytic product-form distributions to enable approximate inferential methods under assumptions of system stationarity, as discussed in the previous example,
\item \textit{generator-matrix} exponentiations with a forward-backward algorithm \citep{zhao2016bayesian} are simply unscalable to such large multivariate systems.
\end{itemize}

\noindent\textbf{Remarks on using benchmark variational methods for inference.} Note that traditional variational methods \citep[cf.][]{opper2008variational,cohn2010mean} are centred around populations or \textit{lengths} in the individual queues. In this example, \textit{populations} \textbf{may not} be factorized under an approximating measure $\mathbb{Q}$, since system jumps are \textit{synchronized}; i.e. a jump \textit{down} in one queue corresponds to a jump \textit{up} in another. As a consequence, pairs of approximating rates under $\mathbb{Q}$ will be interlinked with the same \textit{real} transition rate under $\mathbb{P}$, and derivations such as the lower bound in \eqref{varboundDeveloped}, or equations \eqref{ratesEq}-\eqref{jumpR} are unattainable. For the sake of completeness and comparisons, we adapt existing variational algorithms to the current task; however, we must
\begin{itemize}
\item allow a factorization $\mathbb{Q}(X_t=\boldsymbol{x})=\prod_{i}\mathbb{Q}(X_{i,t}=x_i)$, s.t. jobs may be virtually created and removed in any queue; i.e. jobs do not transition a network, they reach and depart servers individually. The full population of jobs in the network is not preserved,
\item duplicate intensities for transitions in the real model; that is, we have a rate for (i) a job departing a queue, and (ii) the job arriving at another. Technically, a job could arrive at a new server \textbf{before} it departs the previous one; synchronization is lost; point estimates for parameters are averaged across pairs and weighted for network load.
\end{itemize}
We will see, this leads to drastic performance issues that deem the method unusable.

\subsubsection{Algorithm calibration}
Within our method, prior choices in the system state $Y$ must initially accommodate a strictly positive, albeit not necessarily large, likelihood for low-priority jobs to be serviced at any point in time. Here, we achieve this by means of assigning Poisson process priors to task transition counts in $Y$; that is, we first run the master equation \eqref{masterEq} with some user-specified constant intensity rates. This creates monotone mean average queue lengths in the service nodes, and we ensure they flow aligned to the network observations in every instance. 

Also, the presence of strong temporal dependencies will often trigger the approximating rates $\nu$ in \eqref{ratesEq} to become unreasonably large, ultimately deeming the algorithm computationally \textit{unfeasible}. This is a phenomena also observed in \cite{opper2008variational} or \cite{cohn2010mean}, within the context of simpler stochastic dynamics. To ensure computational tractability, we exploit the capping functionals $k$ as in introduced in Proposition \ref{Prop1}, and set a global rate cap of $\bar{\nu}=50$. Furthermore, we run the differential equations for $r$ in \eqref{jumpR} in log-form. Specific details can be found within the aforementioned source code.

\subsubsection{Results}

Within the plots in Figure \ref{Fulltrayectories}, with the exception of the bottom right one, we observe $95\%$ credible intervals for the queue length processes $X^{i,c}_t$ over time, across the various service stations and job classes. There, intervals in dark gray colour relate to high priority jobs, and their corresponding queue length observations are represented by black circles. This information is superimposed over its analogue for low priority jobs, where intervals are coloured in light gray and observations represented by small diamonds. These interval approximations ignore small positive densities that are sometimes assigned to negative queue lengths. Note that this is a consequence of employing counts across job transitions in $Y$ as a basis for inference on $X$, however, we recall this is a necessity in order to overcome the coupling challenges described in Sections \ref{introSec} and \ref{countModelSec}. Overall, we note that the mean-field flow captures well the collected observations, with some few exceptions in the nodes with priority scheduling; hence, it offers a good basis to build approximate estimates for parameters and the likelihood ratio \eqref{likRatio}.

\begin{figure}[h!]
  \centering
    \includegraphics[width=0.98\textwidth]{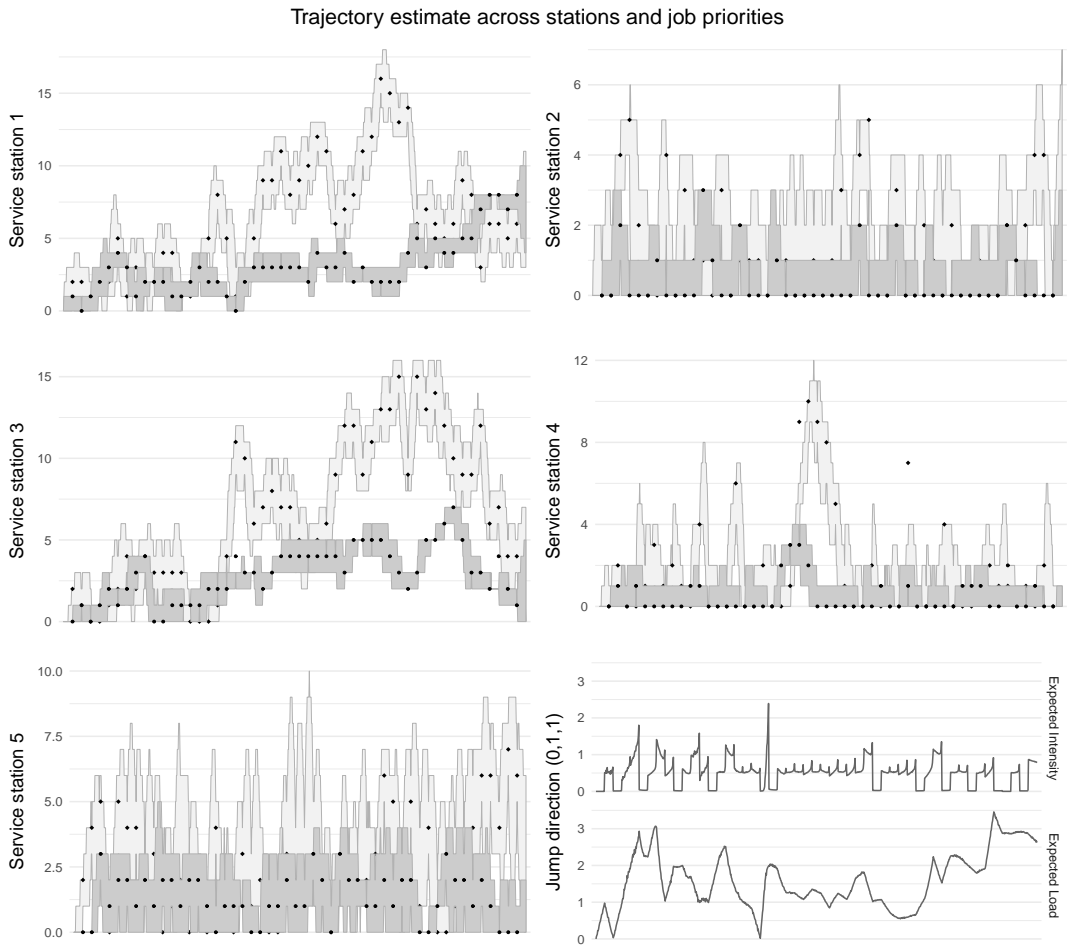}   
   \caption{$95\%$ credible intervals for queue lengths across the service stations. Dark (light) gray corresponds to high (low) priority jobs. Also, expected jump intensity and station load in the direction $\eta=(0,1,1)$. } \label{Fulltrayectories}
\end{figure}

The bottom right plot in Figure \ref{Fulltrayectories} shows an overview of the expected jump intensity $\mathbb{E}^{\mathbb{Q}}_{Y^{\boldsymbol{\eta}}_t} [   \nu^{\boldsymbol{\eta}}_t (Y^{\boldsymbol{\eta}}_t) ]$ and station load $\mathbb{E}^{\mathbb{Q}}_{Y_t} [ \Upsilon(Y_t,\eta_1,\eta_3)\vee 0  ]$ in the direction $\boldsymbol{\eta}=(0,1,1)$ at times $t\in[0,T]$. The sharp peaks in the intensities come at observations times, and ensure the process density transits through the observations. Finally, we notice that the expected station load differs from the estimate of the high-priority queue-length in node $1$, as this process combines and weights the queue-length across the two priorities according to \eqref{Upsil1}.

Next, we find in Table \ref{summaryResults} summary statistics for the posterior service rates under the approximating mean-field measure $\mathbb{Q}$, along with point estimates by adapting benchmark variational techniques in \cite{opper2008variational}. There, we observe how the proposed framework allows for us to gain a good overview of the system properties and variability in the processing speed across the various stations; while existing methods are far from offering reasonable approximations to system behaviour (they instead seem to construct an \textit{averaged} estimation of network flow). Noticeably, there exists a few significant deviations from real values, within the posterior estimates for high priority service rates in \textit{PS} nodes. This is likely due to a combination of sampling variance, high model complexity and the limitations of such approximate variational procedures for transient analyses of stochastic processes.

\begin{table*}[h!]
\caption{Summary statistics for posterior service rates in the queueing network in Figure \ref{exmpl2}.} \label{summaryResults}
\vspace{5pt}
\renewcommand{\arraystretch}{1.1}
\setlength{\tabcolsep}{9pt}
\centering
\begin{tabular}{llllllllll}
\hline
      & Real & O/S & \multicolumn{2}{l}{Summary} & \multicolumn{5}{l}{Quantiles}      \\ \cline{4-10} 
    & &  & Mean        & StDev      & 2.5\%   & 25\%    & 50\%    & 75\%    & 97.5\%     \\ \hline \addlinespace[0.1cm] \hline
$\mu_1^1$ & \textbf{0.25} & 0.364 & 0.307       & 0.043         & 0.228 & 0.276 & 0.304 & 0.335 & 0.397   \\
$\mu_2^1$ & \textbf{1.5} & 1.242 & 1.387      & 0.188         & 1.043 & 1.256 & 1.378 & 1.508 & 1.778       \\
$\mu_3^1$ & \textbf{0.25} & 0.421 & 0.339       & 0.049         & 0.250 & 0.305 & 0.337 & 0.371 & 0.442           \\
$\mu_4^1$ & \textbf{1.5}  & 1.482 & 1.635       & 0.219         & 1.233 & 1.482 & 1.625 & 1.777 & 2.093     \\
$\mu_5^1$ & \textbf{0.5}  & 0.837 & 0.761       & 0.075         & 0.622 & 0.709 & 0.758 & 0.810 & 0.915      \\ \hline \addlinespace[0.1cm] \hline
$\mu_1^2$ & \textbf{0.5}  & 0.551 & 0.501       & 0.041         & 0.424 & 0.473 & 0.499 & 0.528 & 0.584     \\
$\mu_2^2$ & \textbf{4.0}  & 3.496 & 3.740       & 0.298         & 3.177 & 3.534 & 3.731 & 3.935 & 4.346    \\
$\mu_3^2$ & \textbf{0.5}  & 0.568 & 0.504       & 0.041         & 0.425 & 0.475 & 0.502 & 0.531 & 0.588     \\
$\mu_4^2$ & \textbf{4.0}  & 3.265 & 3.670       & 0.296         & 3.112 & 3.465 & 3.661 & 3.863 & 4.270     \\
$\mu_5^2$ & \textbf{1.0}  & 0.976 & 0.984       & 0.056         & 0.877 & 0.946 & 0.983 & 1.021 & 1.097   \\ \hline
\end{tabular}
\end{table*}

\section{Discussion} \label{discussion}

In this paper, we have enabled the variational evaluation of approximating mean-field measures for partially-observed \textit{coupled} systems of jump stochastic processes, with a focus on mixed systems of queueing networks. We furthermore have presented a flexible approximate Bayesian framework, capable of overcoming the challenges posed by coupling properties, and applicable in scenarios where existing MCMC or variational solutions are unusable. To achieve this goal, we have built on existing variational mean-field theory \citep[see][]{opper2008variational,cohn2010mean}, and discussed an alternate optimization procedure with slack variables and inequality constraints that can address computational limitations within existing techniques. Notably, results within this paper contribute to existing Bayesian statistical literature in \cite{sutton2011,Wang2016,perez2017auxiliary}, and first allow for the study of the latent stochastic behaviour across complex mixed network models, by means of an augmented process for interactions in the resources. 

Even though the proposed framework relies on an approximated network model as a basis for inference (which ensures the absolute continuity across base measures), and while it further analyses queue-lengths by means of augmented job transitions in the resources, we have shown we can reliably capture the finite-dimensional posterior distributions of the various marginal stochastic processes, and offer a good overview of the network structure and likely flow of workload. This is important as it can enable the evaluation and uncertainty quantification tasks in several networked systems found in many application domains, where full data observations may be hard to retrieve. Currently, existing state-of-the-art alternatives rely on strong assumptions leading to stationary analyses of such systems, or use alternate MCMC procedures that reportedly find limitations due to existing computational constrains \citep{sutton2011,perez2017auxiliary}. 

\bibliographystyle{apa}     
{\small \bibliography{bibliography}}

\appendix 

\section{Construction} \label{notationApdx}

Let $(\Omega,\mathcal{F})$ be a measurable space with the regular conditional probability property; also, let $0\leq t_1 < \dots < t_K \leq T$ be some fixed observation times, with $T>0$. In a standard queueing network with $M$ stations, $\Omega$ may denote a product set supporting instantaneous rates, trajectories and observations, and $\mathcal{F}$ the corresponding product $\sigma$-algebra. The space of rates and observations will consist of trivial Borel algebras and power sets, so that $\boldsymbol{\lambda}$ is an $(\mathbb{R}_+^n,\mathcal{B}(\mathbb{R}_+^n))$-valued random variable of rates in the infinitesimal generator matrix $Q$ of $X$, where $n\in\mathbb{N}$ denotes an arbitrary number determined by the network topology. In addition, $\{O_k : k=1,\dots,K\}$ corresponds to random measurement variables for the network monitoring activity, each defined on $(\mathcal{O},\mathcal{P}(\mathcal{O}))$, where $\mathcal{O}$ denotes an arbitrary countable support set for observations in every service station. A network trajectory $X=(X_t)_{0\leq t\leq T}$ is an $(\mathcal{S},\mathcal{P}(\mathcal{S}))$-valued stochastic process with a countably infinite support set $\mathcal{S}$. Note that this is a piecewise deterministic jump-process, so that $X=(\boldsymbol{t},\boldsymbol{x})$ is formed by a sequence of transition times $\boldsymbol{t}$ along with states $\boldsymbol{x}$. Every pair $(\boldsymbol{t},\boldsymbol{x})$ can be further defined as a random variable on a measurable space $(\mathcal{X},\Sigma_\mathcal{X})$, with $\mathcal{X}=\cup_{i=0}^\infty ([0,T]\times\mathcal{S})^i$ and the corresponding union $\sigma$-algebra  $\Sigma_\mathcal{X}$. This space can support all finite $\mathcal{S}$-valued trajectories and allows the assignment of a dominating base measure $\mu_\mathcal{X}$ w.r.t which define a trajectory density. For details, we refer the reader to \cite{daley2007introduction}. 

Let $\mathbb{P}$ be a reference probability measure on $(\Omega,\mathcal{F})$. For all $A\in\mathcal{B}(\mathbb{R}^n_+)$, we write $$\mathbb{P}(\boldsymbol{\lambda}^{-1}(A))=\mathbb{P}_{\boldsymbol{\lambda}}(A) =\int_{A}f_{\boldsymbol{\lambda}}(\boldsymbol{a})\,\mu_{\mathbb{R}_+^n}(d\boldsymbol{a}),$$
where $f_{\boldsymbol{\lambda}}$ denotes the joint density function of $n$ independent Gamma distributed variables. Hence, we assume that the distribution of instantaneous rates under $\mathbb{P}$ admits a density carried by a (Lebesgue) measure $\mu_{\mathbb{R}_+^n}$. Next, let $\kappa_1:\mathcal{F}\times\mathbb{R}^n_+\rightarrow[0,1]$ be a regular conditional probability; i.e. a Markov kernel that defines a probability measure on $\mathcal{F}$ for all $\boldsymbol{\lambda}\in\mathbb{R}^n_+$, with
$$\mathbb{P}(B\cap\boldsymbol{\lambda}^{-1}(A))=\int_A \kappa_1(B,\boldsymbol{a}) \, f_{\boldsymbol{\lambda}}(\boldsymbol{a})\,\mu_{\mathbb{R}_+^n}(d\boldsymbol{a})$$
for $A\in\mathcal{B}(\mathbb{R}^n_+)$ and $B\in\mathcal{F}$. By definition, $\kappa_1(B,\boldsymbol{a})= \mathbb{P}(B|\boldsymbol{\lambda}=\boldsymbol{a})$ and most importantly
$$\kappa_1(X^{-1}(C),\boldsymbol{a}) = \int_C f_{X|\boldsymbol{\lambda}=\boldsymbol{a}}(\boldsymbol{t},\boldsymbol{x}) \,\mu_\mathcal{X}(d\boldsymbol{t},d\boldsymbol{x})$$
for all $C\in\Sigma_{\mathcal{X}}$ (note this often poses an intractable integral). The conditional density $f_{X|\boldsymbol{\lambda}=\boldsymbol{a}}$ is such that for every $I\in \mathbb{N}$ and pair of ordered times $\boldsymbol{t}=\{0,t_1,\dots,t_I\}$ in $[0,T]$ and states $\boldsymbol{x}=\{x_0,\dots,x_I\}$ in $\mathcal{S}$ we have
$$
f_{X|\boldsymbol{\lambda}=\boldsymbol{a}}(\boldsymbol{t},\boldsymbol{x}) = \pi(x_0)\,
 e^{Q_{x_{I}}(T-t_I)}\, \prod_{i=1}^I Q_{x_{i-1},x_{i}}\, e^{Q_{x_{i-1}}(t_i-t_{i-1})}, $$
where $Q\equiv Q(\boldsymbol{a})$ is the matrix of infinitesimal transition rates in $X$ associated to values in $\boldsymbol{a}$. Finally, network observations are assumed to be discrete events, independent of transition rates given a trajectory. Thus, there exists a kernel $\kappa_2:\mathcal{F}\times(\mathcal{X} \times \mathbb{R}^n_+)\rightarrow[0,1]$ s.t.
$$\mathbb{P}(O_k\in D | X = (\boldsymbol{t},\boldsymbol{x}), \boldsymbol{\lambda}=\boldsymbol{a}) = \kappa_2(O_k^{-1}(D), (\boldsymbol{t},\boldsymbol{x}),\boldsymbol{a}) = \sum_{d\in D} f_{O_k|(\boldsymbol{t},\boldsymbol{x})}(d) \, \mu_\mathcal{O}(d) $$
for all $k=1,\dots,K$ and $D\in\mathcal{P}(\mathcal{O})$. Here, $f_{O_k|(\boldsymbol{t},\boldsymbol{x})}$ defines an arbitrary probability mass function on $\mathcal{O}$ carried by a counting measure; in our applications, each observation only depends on the state of the system at the observation time, so the above expression could be further simplified. 

Under the above model construction, the support over infinitesimal rates is a standard Borel space and the existence of a posterior distribution is guaranteed (cf. \citet{Orbanz2010}). Also, measures induced by the kernel $\kappa_2$ are $\sigma$-finite and such that $\kappa_2(\cdot,(\boldsymbol{t},\boldsymbol{x}),\boldsymbol{a})<<\mu_{\mathcal{O}}$, for every $((\boldsymbol{t},\boldsymbol{x}),\boldsymbol{a})\in\mathcal{X}\times\mathbb{R}^n_+$. The posterior is thus carried by its corresponding prior and defined by means of the Radon-Nikodym derivative
$$\frac{d\mathbb{P}_{\boldsymbol{\lambda}|O_1=o_1,\dots,O_K=o_K}}{d\mathbb{P}_{\boldsymbol{\lambda}}}(\boldsymbol{a}) = \frac{\int_{\mathcal{X}}\prod_{k=1}^K f_{O_k|(\boldsymbol{t},\boldsymbol{x})}(o_k)\, f_{X|\boldsymbol{\lambda}=\boldsymbol{a}}(\boldsymbol{t},\boldsymbol{x})\, \mu_{\mathcal{X}}(d\boldsymbol{t},d\boldsymbol{x})}{\int_{\mathbb{R}^n_+}\int_{\mathcal{X}}\prod_{k=1}^K f_{O_k|(\boldsymbol{t},\boldsymbol{x})}(o_k)\, f_{X|\boldsymbol{\lambda}=\boldsymbol{a}}(\boldsymbol{t},\boldsymbol{x})\, \mu_{\mathcal{X}}(d\boldsymbol{t},d\boldsymbol{x})\, f_{\boldsymbol{\lambda}}(\boldsymbol{a})\mu_{\mathbb{R}^n_+}(d\boldsymbol{a}) } ,$$
where we employ the shorthand notation $d\mathbb{P}_{\boldsymbol{\lambda}|\cdot}(\boldsymbol{a}) = \mathbb{P}_{\boldsymbol{\lambda}}(d\boldsymbol{a}|\cdot)$.

\end{document}